\documentclass[UKenglish,cleveref, autoref, thm-restate,a4paper]{article}
\pdfoutput=1
\usepackage{arxiv}
\usepackage{multirow}
\usepackage{url}
\usepackage{hyperref}
\hypersetup{
    colorlinks=true,
    linkcolor=blue,
    filecolor=magenta,      
    urlcolor=red,
    citecolor=blue,
    pdftitle={Overleaf Example},
    pdfpagemode=FullScreen,
    }
    
\usepackage[margin=1in]{geometry} 
\usepackage{amssymb,amsmath,amsthm, amsfonts}
\usepackage{qcircuit}
\usepackage{thmtools}
\usepackage{thm-restate}
\usepackage{tikz}
\usepackage{cleveref}
\usepackage{float}
\usepackage{graphicx}
\usepackage{float}
\usepackage{color}
\usepackage{braket}
\usepackage{xcolor}
\usepackage{tablefootnote}
\usepackage{graphicx}
\usepackage{amsmath}
\usepackage{amssymb}
\usepackage{appendix}
\usepackage{amsthm}
\usepackage{amsfonts}
\usepackage{tcolorbox}
 \tcbuselibrary{skins , breakable , listings ,theorems}
\usepackage{colortbl}
\usepackage{comment}
\usepackage{booktabs}
\usepackage{bbm}
\usepackage{amsfonts}
\usepackage{graphicx}
\usepackage{amsmath}
\usepackage{amsthm}
\usepackage{braket}
\usepackage{mathrsfs}
\usepackage{mathtools}
\usepackage{verbatim}
\newcommand{\B}{\{0,1\}}
\newcommand{\GHD}{\mathrm{GHD}}
\newcommand{\UDISJ}{\text{\sc UDisj}}

\newcommand{\eGHD}{\overline{\GHD}}

\newcommand{\logGamma}[1]{\mathrm{\log\widetilde{\gamma}_2}}

\newcommand{\dom}[1]{\mathrm{dom}}
\renewcommand{\Re}{\mathbb{R}}
\DeclareMathOperator{\adeg}{\mathsf{\widetilde{deg}}}

\DeclareMathOperator{\arank}{\widetilde{rank}}
\DeclareMathOperator{\rank}{\mathrm{rank}}

\DeclareMathOperator{\DISJ}{\mathsf{DISJ}}
\DeclareMathOperator{\Equality}{\mathsf{EQ}}
\newcommand{\inbrace}[1]{\left \{ #1 \right \}}
\newcommand{\inparen}[1]{\left ( #1 \right )}

\newcommand{\abs}[1]{\lvert #1 \rvert}

\newcommand{\ESetInc}{\mathsf{ESetInc}}
\newcommand{\SetInc}{\mathsf{SetInc}}

\newcommand{\bit}{\set{0,1}}

\newcommand{\domain}{\mathcal{X}}

\newcommand{\sign}{\mathrm{sign}}

\newcommand{\CC}{\mathrm{CC}}

\newcommand{\defemph}[1]{\textbf{\emph{#1}}}

\newtheorem{theorem}{Theorem}[section]
 
\newtheorem{lemma}[theorem]{Lemma}

\newtheorem{definition}[theorem]{Definition}
\newtheorem{remark}[theorem]{Remark} 
 
\newtheorem{fact}[theorem]{Fact}
   
\newtheorem{conjecture}[theorem]{Conjecture}

\title{Quantum and Classical Communication Complexity of Permutation-Invariant Functions}

\date{}

\author{Ziyi Guan\thanks{Partially supported by the Ethereum Foundation.}\\ 	EPFL, Lausanne, Switzerland\\ \textit{ziyi.guan@epfl.ch} 
\And
Yunqi Huang\\University of Technology Sydney, Sydney, Australia\\
\textit{	yunqi.huang@student.uts.edu.au}
\And
	Penghui Yao\thanks{Supported by the National Natural Science Foundation of China under Grant
62332009 and Grant 12347104, the Innovation Program for
Quantum Science and Technology under Grant 2021ZD0302901, the National Natural Science Foundation of China (NSFC)/Research
Grants Council of Hong Kong (RGC) Joint Research Scheme under
Grant 12461160276, the Natural Science Foundation of Jiangsu
Province under BK20243060, and the New Cornerstone Science
Foundation.}\\ State Key Laboratory for Novel Software Technology, New Cornerstone Science Laboratory, Nanjing University, China \\ Hefei National Laboratory, Hefei 230088, China \\ \textit{phyao1985@gmail.com} 
\And
Zekun Ye\thanks{Supported by the National Natural Science Foundation of China under Grant
62332009 and Grant 12347104, the Innovation Program for
Quantum Science and Technology under Grant 2021ZD0302901, the National Natural Science Foundation of China (NSFC)/Research
Grants Council of Hong Kong (RGC) Joint Research Scheme under
Grant 12461160276, the Natural Science Foundation of Jiangsu
Province under BK20243060, and the New Cornerstone Science
Foundation.}\\ State Key Laboratory for Novel Software Technology, New Cornerstone Science Laboratory, Nanjing University, China\\ \textit{yezekun@smail.nju.edu.cn} 
}

\begin{document}
\maketitle

\begin{abstract}
This paper gives a nearly tight characterization of the quantum communication complexity of permutation-invariant Boolean functions. With such a characterization, we show that the quantum and randomized communication complexity of permutation-invariant Boolean functions are quadratically equivalent (up to a polylogarithmic factor {of the input size}). Our results extend a recent line of research regarding query complexity %\cite{AA14, Cha19, BCG+20} 
to communication complexity, showing symmetry prevents exponential quantum speedups. Furthermore, we show that the Log-rank Conjecture holds for any non-trivial total permutation-invariant Boolean function. Moreover, we establish a relationship between the quantum/classical communication complexity and the approximate rank of permutation-invariant Boolean functions. This implies the correctness of the Log-approximate-rank Conjecture for permutation-invariant Boolean functions in both randomized and quantum settings (up to a polylogarithmic factor {of the input size}).
\end{abstract}

%\begin{IEEEkeywords}
%communication complexity, Log-rank Conjecture, permutation-invariant functions, quantum advantages.
%\end{IEEEkeywords}

\section{Introduction}

%\IEEEPARstart{E}{xploring} 
Exploring quantum advantages is a key issue in the realm of quantum computing. Numerous work focuses on analyzing and characterizing quantum advantages, such as  \cite{ATYY17,BGKT19, GS20,CCHL21,Kallaugher21,YZ22}. It is known that quantum computers {can} demonstrate a potential exponential speedup compared with classical computers, such as Simon's problem \cite{simon1994power} and integer factoring  \cite{shor1994algorithms}, for which quantum algorithms exploit the internal structures of the problems (e.g. fast Fourier transform). However, for some highly unstructured problems, such as the unstructured search \cite{Grover96} and collision problems \cite{AS04}, quantum speedups are at most polynomial. In light of the aforementioned phenomenon, Aaronson and Ambainis \cite{AA14} asked: \emph{How much structure is needed for huge quantum speedups?}
%How much structure is needed for huge quantum speedups?

Regarding the above {question}, there are two major directions to explore the structure needed for quantum speedups in the query model, a complexity model commonly used to describe quantum advantages. On the one hand, Aaronson and Ambainis \cite{AA14} conjectured the acceptance probability of a quantum query algorithm to compute a Boolean function can be approximated by a classical deterministic algorithm with only a polynomial increase in the number of queries, which is one of most important conjecture in the field of Boolean analysis. On the other hand,  
Watrous conjectured that quantum and randomized query complexities are polynomially equivalent for any permutation-invariant function \cite{AA14}. Along this direction, Aaronson and Ambainis \cite{AA14} initiated the study on the quantum speedups of permutation-invariant functions for query complexity. They demonstrated that any (partial) function that is invariant under full symmetry does not exhibit exponential quantum speedups, thereby resolving the Watrous conjecture.  %(Interested readers may refer to \cite{AA14} for a more detailed introduction.) %{Here, a partial function is a function defined only on a subset of its domain $\domain$. Correspondingly, a total function is a function that is defined on the entire domain.}
Furthermore, Chailloux \cite{Cha19} expanded upon their work by providing a tighter bound and removing a technical constraint. Recently, Ben-David, Childs, Gilyén, Kretschmer, Podder, and Wang \cite{BCG+20} further proved that hypergraph symmetries in the adjacency matrix model allow at
most a polynomial separation between randomized and quantum query complexities.
All the above results demonstrated that symmetries break exponential quantum speedups in the query model.

The study of the roles of ``structure'' in quantum speedups has obtained considerable attention in the query model, which leads us to consider whether we can derive similar results in other computation models. {In this paper, we study the communication complexity model introduced by Yao~\cite{yao1979some}. In the model of communication complexity, the inputs are distributed among two separated parties, each party is assumed to be computationally unbounded. The communication complexity studies the minimum number of bits the players need to exchange to achieve a task. In the model of quantum communication complexity~\cite{yao1993quantum}, the players are allowed to exchange quantum messages. Quantum communication complexity comes to attention as it is also extensively used to demonstrate quantum advantages.} Furthermore, while the exponential gap between quantum and classical communication models has been shown in many works \cite{raz1999exponential,bar2008exponential,GKKRW07,GP08,Montanaro11}, 
there are also some problems in communication models that demonstrate at most polynomial quantum speedups, such as the Set-Disjointness problem \cite{razborov2003quantum} and the (gap) Hamming-Distance-Problem \cite{HSZZ06,She12,Vid12,CR12}.   
Therefore, it is intriguing to consider how much structure is needed for significant quantum speedups in the communication complexity model. More specifically, would symmetry also break quantum exponential advantages in the communication complexity model? In this paper, we investigate a variant of the Watrous conjecture concerning the quantum and randomized communication complexities of permutation-invariant functions  (\Cref{conjecture:cc_version_watrous}). Briefly, a permutation-invariant Boolean function is a function that is invariant under permutations of its inputs. We provide the formal definition below. %Specifically, the formal definition is as follows. 
%(see \Cref{def:PI} for a formal definition).
%In the two-party communication model, the function value of a permutation-invariant function is invariant if we perform the same permutation to the inputs of Alice and Bob. 
\begin{definition}[Permutation-invariant (PI) functions~\cite{GKS16}]
\label{def:PI}
A (total or partial) function $f:\{0,1\}^n \times \{0,1\}^n \to \{-1,1,*\}$ is permutation-invariant if for all $x,y \in \{0,1\}^n$, and every bijection $\pi:\inbrace{0,...,n-1} \to \inbrace{0,...,n-1}$, $f(x^{\pi}, y^{\pi}) = f(x,y)$, where $x^\pi$ satisfies that $x^{\pi}_{(i)} = x_{\pi(i)}$ for any $i \in \inbrace{0,...,n-1}$.
\end{definition}
\begin{remark}\label{remark:and}
    {For permutation-invariant functions, the function value depends only on the joint type of the input. Specifically, any permutation-invariant function $f$ in \Cref{def:PI} depends only on $|x|, |y|$ and $|x \land y|$. %(also depends only on $|x|,|y|$ and $\Delta(x,y)$). 
    Here $|\cdot|$ is the Hamming weight of the binary string, i.e., the number of 1's in the string, and $|x \land y| $ is the number of $i$ such that $x_i = y_i = 1$.} %Moreover, $\Delta(x,y) = |\inbrace{i\in [n]|x_i \neq y_i}|$ is the Hamming distance of } 
\end{remark}

\begin{conjecture}[Communication complexity version of the Watrous Conjecture]\label{conjecture:cc_version_watrous} 
{Fix $m \in \mathbb{Z^+}$.} For any permutation-invariant function $f:\inbrace{0,1,...,m}^n \times \inbrace{0,1,...,m}^n \rightarrow \inbrace{-1,1,*}$, $R(f) \le Q(f)^{O(1)}$, where $R(f)$ and $Q(f)$ are the randomized and quantum communication complexities of $f$, respectively.
\end{conjecture}

Furthermore, we study the Log-rank Conjecture proposed by Lovasz and Saks \cite{LS88}, a long-standing open problem in communication complexity. Despite its difficulty on total functions~\cite{10.1145/2629598,10.1145/2724704, 10.1145/3406325.345099}, the conjecture has been shown for several subclasses of total permutation-invariant Boolean functions~\cite{BW01} and $\mathrm{XOR}$-symmetric functions \cite{ZS09}. Lee and Shraibman \cite{LS09an} further proposed the  Log-approximate-rank Conjecture, stating that the randomized communication and the logarithm of the approximate rank of the input matrix are polynomially equivalent. However, this conjecture was later proven false \cite{10.1145/3396695}, even for its quantum analogue~\cite{8948623,8948635}.

In this paper, we investigate both conjectures for permutation-invariant functions.

\begin{conjecture}[Log-rank Conjecture for permutation-invariant functions] 
For any total permutation-invariant function $f:\B^n \times \B^n \rightarrow \inbrace{-1,1}$, $D(f) \le \inparen{\log\rank(f)}^{O(1)}$, where $\rank(f)$ is the rank of the input matrix of $f$.
\end{conjecture}

\begin{conjecture}[Log-Approximate-Rank Conjecture for permutation-invariant functions] 
For any (total or partial) permutation-invariant function $f:\B^n \times \B^n \rightarrow \inbrace{-1,1,*}$, $R(f) \le \inparen{\log\arank(f)}^{O(1)}$, where $\arank(f)$ is the approximate rank of the input matrix of $f$ (see \Cref{def:arank}). % for a formal definition).
\end{conjecture}

\begin{conjecture}[Quantum Log-Approximate-Rank Conjecture for permutation-invariant functions] 
For any (total or partial) permutation-invariant function $f:\B^n \times \B^n \rightarrow \inbrace{-1,1,*}$, $Q(f) \le  \inparen{\log\arank(f)}^{O(1)}$.
\end{conjecture}

%%%%%%%%%%%%%%%%%%%%%%%%%%%%%%%%%%%%%%%%%%%%%%%%%%%%%%%%%%%%%%%%%%%%%%%%%%%%%%%%%
%%%%%%%%%%%%%%%%%%%%%%%%%%%%%%%%%%%%%%%%%%%%%%%%%%%%%%%%%%%%%%%%%%%%%%%%%%%%%%%%
%\subsection{Rank and Approximate Rank}
%If $F$ is a real matrix, let $\rank(F)$ be the \defemph{rank} of $F$. Then we define the approximate rank for any incomplete matrix as follows.

\begin{definition}[Approximate rank]\label{def:arank}
    For matrix $M \in \inbrace{-1,1,*}^{m \times n}$ (`$*$' means undefined entry) and $0\leq \epsilon < 1$, we say a real matrix $A$ approximates $M$ with error $\epsilon$ if:
\begin{enumerate}
\item[1)]  $|A_{i,j}-M_{i,j}|\leq \epsilon$ for any $i \in [m], j\in [n]$ such that $M_{i,j} \neq *$;
\item[2)]  $|A_{i,j}|\leq 1$ for all $i\in [m], j\in [n]$.
\end{enumerate}
Let $\mathcal{M}_\epsilon$ be the set of all the real matrices that approximate $M$ with error $\epsilon$. The approximate rank of $M$ with error $\epsilon$, denoted by $\arank_{\epsilon}(M)$, is the smallest rank among all real matrices in $\mathcal{M}_{\epsilon}$. If $\epsilon = 2/3$, we abbreviate $\arank_{\epsilon}(M)$ as $\arank(M)$. For any {(total or partial)} Boolean function {$f:\{0,1\}^n\times\{0,1\}^n \rightarrow \{-1,1,*\}$}, let 
%$\rank(f)\coloneqq \rank\inparen{M_f}$ and 
$\arank(f)\coloneqq \arank{\inparen{M_f}}$, where $M_f$ is the input matrix of $f$. 
\end{definition}

\subsection{Our Contribution}
To study the communication complexity version of the Watrous conjecture, we start with permutation-invariant Boolean functions, which is an important step towards fully resolving \Cref{conjecture:cc_version_watrous}. %essential to analyze general permutation-invariant functions {with non-Boolean input as \Cref{conjecture:cc_version_watrous}}. 
We show that for every permutation-invariant Boolean function, its classical communication complexity has at most a quasi-quadratic blowup compared to its quantum communication complexity (\Cref{th:main}).
Thus, we cannot hope for significant quantum speedups of permutation-invariant Boolean functions. Additionally, \Cref{th:main} gives a nearly tight bound on the quantum communication complexity {up to a polylogarithmic factor of the input size}. Furthermore, we show that every non-trivial permutation-invariant Boolean function satisfies the Log-rank Conjecture in \Cref{th:logrank}. To resolve the (quantum) Log-Approximate-Rank Conjecture, 
we investigate the relationship between the quantum/classical communication complexities and the approximate rank of any permutation-invariant Boolean function in \Cref{th:PI_log_arank}.

Consider a Boolean function $f$.
Let $D(f),R(f),Q(f)$ be the deterministic communication complexity, randomized communication complexity, and quantum communication complexity of $f$, respectively. Let $\rank\inparen{f}$ and $\arank\inparen{f}$ be the rank and approximate rank of $f$. 
We summarize our results 
below. 
%{where \Cref{th:main} corresponds to \Cref{tab:my_label}, and \Cref{th:logrank,th:PI_log_arank} corresponds to \Cref{tab:protocols}.} %\footnote{In \Cref{th:main,th:PI_log_arank}, $\widetilde{O}\inparen{M\inparen{f}} = O\inparen{M(f)\log^2 n\log \log n}$ for any complexity measure $M$.}.
\begin{restatable}{theorem}{maintheorem}
\label{th:main}
    For any (total or partial) permutation-invariant function $f:\{0,1\}^n \times \{0,1\}^n \to \{-1,1,*\}$ in Definition \ref{def:PI}, the followings hold:
\[
\begin{aligned}
&\Omega\inparen{m\inparen{f}} \le R(f)  \le {O\inparen{m\inparen{f}^2\log^2 n \log\log n+\log n}},\\ &\Omega\inparen{m\inparen{f}}  \le Q(f)   \le {O\inparen{m\inparen{f}\log^2 n \log\log n+\log n}},\\
%&\Omega\inparen{m\inparen{f}} \le R(f) \le \widetilde{O}\inparen{m\inparen{f}^2} \le \widetilde{O}\inparen{Q(f)^2}
%\text{ and }\\ &\Omega\inparen{m\inparen{f}}  \le Q(f) \le Q(f) \le \widetilde{O}\inparen{m\inparen{f}},\\
\end{aligned}
\]
where $m(f)$ is a measure defined in Definition \ref{def:mf}. Hence, $R(f)\le {O(Q(f)^2\log^2 n \log\log n+\log n)}$ for any permutation-invariant function $f$.
\end{restatable}

\begin{restatable}{theorem}{logrank}\label{th:logrank}
    For any non-trivial total permutation-invariant function $f:\{0,1\}^n \times \{0,1\}^n \to \{-1,1\}$ in \Cref{def:PI}, we have 
    \[
    D(f) = O\inparen{\log^2\rank\inparen{f}}.
    \]
    Here, we say $f$ is non-trivial if $f(x,y)$ does not only depend on $|x|$ and $|y|$.
\end{restatable}

\begin{restatable}{theorem}{PIlogarank}\label{th:PI_log_arank}
    For any (total or partial) permutation-invariant function $f:\{0,1\}^n \times \{0,1\}^n \to \{-1,1,*\}$ in \Cref{def:PI}, we have 
    \[
    \begin{aligned}
    R(f) &= {O\inparen{\log^2 \arank\inparen{f}\log^2n \log \log n+\log n}},\\
    Q(f) &= {O\inparen{\log \arank\inparen{f}\log^2 n \log\log n+\log n}}.\\
    %R(f) = \widetilde{O}\inparen{\log^2 \arank\inparen{f}}~\text{ and}~
    %Q(f) = \widetilde{O}\inparen{\log \arank\inparen{f}}.\\
    \end{aligned}
    \]
\end{restatable}

\begin{remark}
The relations between $R(f),Q(f),\log\arank(f)$ are tight in \Cref{th:main,th:PI_log_arank} (up to a polylogrithmic factor of the input size), since for the Set-Disjointness Problem, we have $R(f) = \Omega(Q(f)^2)$ and $Q(f) = \Omega\inparen{\log\arank(f)}$ \cite{lee2009lower}.  
\end{remark}

% \begin{table}[h]
% \caption{Quantum/Classical polynomial equivalence of PI functions}
%     \centering \resizebox{0.6\textwidth}{!}{
%     \begin{tabular}{cc}
%     \toprule
%     Query model & \textbf{Communication model (this work)} \\
%     \midrule
%      %QC polynomial equivalence & 
%      True~\cite{AA14,Cha19}    
%          & \textbf{True} (up to a polylogarithmic factor of the input size)  \\
%     \bottomrule
%     \end{tabular}
%     }
%     \label{tab:my_label}
% \end{table}

% \begin{table}[h]
% \caption{Conjectures in communication complexity}
% \centering
% \resizebox{\textwidth}{!}{
% \begin{tabular}{ccc}
% \toprule
%  & Total functions & \textbf{PI functions (this work)}   \\
% \midrule
% Log-rank Conjecture \cite{LS88} & Unknown & \textbf{True} (non-trivial total PI functions) \\
% %\midrule
% Log-approximate-rank Conjecture \cite{LS09an} & False~\cite{10.1145/3396695}  & \textbf{True} (up to a polylogarithmic factor of the input size) \\
% %\midrule
% Quantum Log-approximate-rank  Conjecture \cite{LS09an} & False~\cite{8948623,8948635} & Same as above \\
% \bottomrule
% \end{tabular}
% }
% \label{tab:protocols}
% \end{table}

\subsection{Related Work}
\begin{comment}
\begin{table}[h]
\caption{Three Communication Models}
\label{tab:protocols}
\centering
\resizebox{1.0\textwidth}{!}{
\begin{tabular}{cccc}
\toprule
& \textbf{Communication Complexity} & \textbf{Medium} & \textbf{Randomness} \\%&\textbf{Error allowed} \\
\midrule
Deterministic & $\mathrm{DCC}(f)$ & Bits & Without randomness \\ %& Without error\\
\midrule
Randomized & $\mathrm{RCC}(f)$ &  Bits & Public randomness \\ % & Bounded error\\
\midrule
Quantum & $\mathrm{QCC}(f)$ & Qubits & Prior entanglement \\ %& Bounded error\\
\bottomrule
\end{tabular}
}
\end{table}
\end{comment}
The need for structure in quantum speedups has been studied in the query model extensively. Beals, Buhrman, Cleve, Mosca and de Wolf \cite{BBC+01} demonstrated that there exists at most a polynomial quantum speedup for total Boolean functions in the query model. Moreover, Aaronson and Ambainis \cite{AA14} established that partial symmetric functions also do not allow super-polynomial quantum speedups. Chailloux \cite{Cha19} further improved this result to a broader class of symmetric functions. Ben-David, Childs, Gilyén, Kretschmer, Podder and Wang \cite{BCG+20} later analyzed the quantum advantage for functions that are symmetric under different group actions systematically. Ben-David \cite{BenDavid16} established a quantum and classical polynomial equivalence for a certain set of functions satisfying a specific symmetric promise. Aaronson and Ben-David \cite{AB16} proved that if domain $D$ satisfies $D = \text{poly}(n)$, there are at most polynomial quantum speedups for computing an $n$-bit partial Boolean function. 

In terms of communication complexity, there are a few results that imply the polynomial equivalence between quantum and classical communication complexity for several instances of permutation-invariant functions. 
Examples include $\mathrm{AND}$-symmetric functions~\cite{razborov2003quantum}, Hamming-Distance problem~\cite{HSZZ06,CR12}, $\mathrm{XOR}$-symmetric functions~\cite{ZS09}.
While the above results characterized quantum advantage for a certain class of permutation-invariant Boolean functions, our work provides a systemic analysis of all permutation-invariant Boolean functions.  

The study of the Log-rank Conjecture and the Log-Approximate-Rank Conjecture has a rich history. Here, we survey the results of the Log-rank Conjecture and the Log-Approximate-Rank Conjecture about permutation-invariant Boolean functions. {Buhrman} and de Wolf \cite{BW01} verified the correctness of the Log-rank Conjecture for $\mathrm{AND}$-symmetric functions.
Combining the results of Razborov \cite{razborov2003quantum}, Sherstov \cite{She11} and Suruga \cite{Sur2023},  
it is implied that the Log-Approximate-Rank Conjecture holds for $\mathrm{AND}$-symmetric functions both in the randomized and quantum settings. Moreover, Zhang and Shi~\cite{ZS09} showed $\mathrm{XOR}$-symmetric functions satisfy the Log-Rank Conjecture. Chattopadhyayand and Mande \cite{CM17} further proved that the Log-Approximate-Rank Conjecture holds for $\mathrm{XOR}$-symmetric functions for the first time. 

\subsection{Proof Techniques}
In this section, we give a high-level technical overview of our main results.

%\subsubsection{Lower Bound}
%\label{sec:tech-lower-bound}
First, we outline our approaches to 
obtain the lower bound on the quantum communication complexity,
%(\Cref{{th:low}}), 
rank %(\Cref{{log_rank_lower_bound}}) 
and approximate rank of permutation-invariant functions below:
\begin{enumerate}
    \item \textbf{Quantum communication complexity} and \textbf{approximate rank}: In \Cref{th:main}, %to prove $Q(f) = \Omega(m(f))$ for any permutation-invariant function $f$, 
    we use the following two-step reduction (see {Theorem \ref{th:low}} and \Cref{th:qcc_lowbound}): First, we reduce the lower bound of any permutation-invariant function to the lower bound of Exact Set-Inclusion Problem ($\mathrm{ESetInc}$, a specific instance of permutation-invariant functions, \Cref{def:setinc}). Second, we reduce the lower bound of the quantum communication complexity of $\mathrm{ESetInc}$   to Paturis's approximate degree of symmetric functions \cite{Paturi92} by the pattern matrix method \cite{She11}, a well-known method for lower bound analysis in quantum communication complexity.     In \Cref{th:PI_log_arank}, we use a similar method to prove the lower bound of approximate rank. %: $\log \arank\inparen{f} = \Omega(m(f))$.

    \item \textbf{Rank}: In \Cref{th:logrank}, we reduce
    the lower bound of the rank of total permutation-invariant functions to the lower bound of the rank of some representative function instances, such as the Set-Disjointness Problem and the Equality Problem (see \Cref{log_rank_lower_bound}).
\end{enumerate}

%\subsubsection{Upper Bound}
Moreover, we use the following methods to show the upper bounds on the communication complexity of permutation-invariant functions in the randomized, quantum, and deterministic models.

\begin{enumerate}
    \item \label{tech_random_upper_bound}
    \textbf{Randomized and quantum models}: In \Cref{th:main}, to prove {the randomized upper bound} for permutation-invariant functions, we first propose a randomized protocol to solve the Set-Inclusion problem ($\mathrm{SetInc}$, \Cref{def:setinc}) using a well-suited sampling method according to the parameters of $\mathrm{SetInc}$ (see \Cref{le:extiXandY}). 
    Afterwards, we use this protocol as a subroutine to solve any permutation-invariant function based on binary search (see \Cref{th:upper}). Furthermore, to prove {quantum upper bound}, we use the quantum amplitude amplification technique \cite{brassard2002quantum,HM19} to speed up the above randomized protocol to solve $\mathrm{SetInc}$ (see \Cref{lemma:quantum_upper_bound_setinc}). {It is also worth noting that in the quantum protocols, Alice and Bob do not need to share prior entanglement.}

    \item \textbf{Deterministic model}: In \Cref{th:logrank}, to give an upper bound on the deterministic communication complexity of total permutation-invariant functions, we propose a deterministic protocol as follows (see \Cref{log_rank_upper_bound}): Alice and Bob first share the Hamming weight of their inputs, and decide who sends the input to the other party. %according to the definition of function and the Hamming weight of inputs. 
        The party that has all the information about inputs will output the answer. %Combining Lemmas \ref{log_rank_lower_bound} (described in \Cref{sec:tech-lower-bound}) and \ref{log_rank_upper_bound}, \Cref{th:logrank} can be proved.
    \end{enumerate}

\subsection{{Comparison with previous work}}
In previous work, Ghazi, Kamath and Sudan \cite{GKS16} considered 
the randomized communication complexity of permutation-invariant functions. For any permutation-invariant function $f$, they introduced a complexity measure {(denoted by $M(f)$)} almost equivalent up to a fourth power of $R(f)$, i.e.,
\[
\begin{aligned}
R(f) &= \Omega(M(f)),\\
R(f) &= 
O\inparen{M(f)^4\log M(f)\log\log n \log\log\log n+\log n}.
\end{aligned}
\]
In this paper, we propose a new complexity measure $m(f)$, refining the argument in \cite{GKS16}, which is {almost} quadratically equivalent to $R(f)$ and tightly characterizes $Q(f)$ {up to a polylogarithmic factor of the input size} (See \Cref{th:main}). This enables us to prove a quadratic equivalence between $R(f)$ and $Q(f)$ for all permutation-invariant functions $f$.

 %The similarities between \cite{GKS16} and our work are shown as follows. 
%For any permutation-invariant functions $f$, $f(x,y)$ only depends on $|x|,|y|,\Delta(x,y)$ or $|x|,|y|,|x \land y|$, 

Specifically, \cite{GKS16} first considered the Gap-Hamming-Distance Problem as \Cref{def:GHD}, and then generalized the result to all permutation-invariant functions. %Additionally, $M(f)$ is the randomized lower bound of all subfunctions of $f$.
In our work, we use a similar framework. The difference point is that, instead of the Gap-Hamming-Distance Problem, we consider the Set-Inclusion Problem as \Cref{def:setinc}. It is worth noting that, while the two definition forms are different, the Set-Inclusion Problem is equivalent to the Gap-Hamming-Distance Problem, i.e., $\SetInc_{a,b,c,g}^n$ is equivalent to $\GHD_{a,b,a+b-2c,2g}^n$.
 %Specifically, the Set-Inclusion Problem instance $\SetInc_{a,b,c,g}^n$ is to distinguish $|x \land y| = c-g$ from $c+g$ provided $|x| = a, |y| = b$. It also can be defined as $\GHD_{a,b,a+b-c,g}^n$ equivalently, i.e. to distinguish $\Delta(x,y) = a+b-c+g$ from $a+b-c-g$ provided $|x| = a, |y| = b$. 

%For the upper bound, 

\begin{definition}[Gap-Hamming-Distance Problem \cite{GKS16}]\label{def:GHD}
{Fix $n \in \mathbb{Z}^+$. Consider $a,b\in \inbrace{1,...,n-1}$ and $c-g,c+g$ are achievable Hamming distances of $\Delta(x,y)$ when $|x| = a, |y| = b$.} The Gap-Hamming-Distance Problem $\GHD^n_{a,b,c,g}:\inbrace{0,1}^n\times\inbrace{0,1}^n\rightarrow\inbrace{-1,1,*}$ is defined as the following partial function:
\begin{align*}
&\GHD_{a, b, c, g}^{n}(x, y)\\
&\coloneqq\begin{cases}
-1 &\text{if $|x|=a,|y|=b$  and  $\Delta(x, y)\ge c+g$}, \\
1 &\text{if $|x|=a,|y|=b$  and  $\Delta(x, y)\le c-g$}, \\
* &\text{otherwise,}
\end{cases}
\end{align*}
where $\Delta(x,y) = |\inbrace{i \in \{0,1,...,n-1\}|x_i \neq y_i}|$ is the Hamming distance between $x$ and $y$. Additionally, the Exact Gap-Hamming-Distance Problem $\eGHD_{a, b, c, g}^{n}$ as follows. 
\begin{align*}
&\eGHD_{a, b, c, g}^{n}(x, y)\\
&\coloneqq\begin{cases}
-1 &\text{if $|x|=a,|y|=b$  and  $\Delta(x, y) = c-g$}, \\
1 &\text{if $|x|=a,|y|=b$  and  $\Delta(x, y) = c+g$}, \\
* &\text{otherwise.}
\end{cases}
\end{align*}
\end{definition}

\begin{definition}[Set-Inclusion Problem]\label{def:setinc}
{Fix $n \in \mathbb{Z}^+$. Consider $a,b\in \inbrace{1,...,n-1}$ and $c-g,c+g$ are achievable Hamming weights of $|x \land y|$ when $|x| = a, |y| = b$.} The Set-Inclusion Problem $\mathrm{SetInc}^n_{a,b,c,g}:\inbrace{0,1}^n\times\inbrace{0,1}^n \rightarrow \inbrace{-1,1,*}$ is defined as the following partial function:
\begin{align*}
&\mathrm{SetInc}_{a, b, c, g}^{n}(x, y)\\
&\coloneqq\begin{cases}
-1 &\text{if $|x|=a,|y|=b$  and  $|x \land y|\le c-g$}, \\
1 &\text{if $|x|=a,|y|=b$  and  $|x \land y|\ge c+g$}, \\
* &\text{otherwise.}
\end{cases}
\end{align*}
Additionally, the Exact Set-Inclusion Problem $\mathrm{ESetInc}_{a, b, c, g}^{n}$ is defined as follows. 
\begin{align*}
&\mathrm{ESetInc}_{a, b, c, g}^{n}(x, y)\\
&\coloneqq\begin{cases}
-1 &\text{if $|x|=a,|y|=b$  and  $|x \land y|=c+g$}, \\
1 &\text{if $|x|=a,|y|=b$  and  $|x \land y|=c-g$}, \\
* &\text{otherwise.}
\end{cases}
\end{align*}
\end{definition}

 % and the corresponding upper bounds are shown as \Cref{tab:comparison}. 
 
 Next, we show the difference between \cite{GKS16} and our result about the Set-Inclusion Problem (or Gap-Hamming-Distance Problem equivalently) from upper and lower bounds respectively. 
 
 For the upper bound, we designed different protocols from \cite{GKS16}. The details of our protocols are shown in \ref{subsec:upper-bound}. Moreover, the comparison of required communication costs is shown in \Cref{tab:comparison}.
 
For the lower bound, %the schematic diagram of the reduction map of  \cite{GKS16} is shown as \Cref{fig:reduction map}. 
noting that the information complexity of the Unique-Set-Disjointness problem ($\mathrm{UDISJ}$) has been known \cite{bar2008exponential},
by a series of reductions,  \cite{GKS16} showed the information complexity of $\UDISJ$ is a lower bound on the randomized communication complexity of the Exact Gap-Hamming-Distance Problem (See Section 3.3 in \cite{GKS16}). As a comparison, we use a different reduction method as Figure \ref{fig:reduction_map}. 
\tikzstyle{class}=[]
\tikzstyle{arrow} = [->]
\begin{figure}[htbp]
\begin{center}
\begin{tikzpicture}[node distance=0cm]
\node[class](rootnode){$\mathrm{ESetInc}_{a, b, c, g}^{n}$};
\node[class,below of=rootnode,yshift=0,xshift=5.5cm](first-1){$\mathrm{ESetInc}_{n_1+n_2,n_1+n_2,n_1,g}^{n_1+3n_2}$};
\node[class,below of=first-1,yshift=-2.5cm,xshift=0cm](two-1)
{$\mathrm{ESetInc}_{2k,k,l,1/2}^{4k}$};
%{$\adeg(f_{m_1,m_2})$};
\draw [arrow] (rootnode) -- node [font=\small,midway, above] {\Cref{le:qcc_lowbound1}} (first-1);
\draw [arrow] (first-1) -- node [font=\small,midway,left] {\Cref{le:qcc_lowbound2}} (two-1);
\end{tikzpicture}
\caption{Our reduction path for the lower bound of $\ESetInc$.}
\label{fig:reduction_map}
\end{center}
\end{figure}
We first reduce the lower bound of the Exact Set-Inclusion Problem to some smaller instances with specific parameters (See \Cref{le:qcc_lowbound1,le:qcc_lowbound2}). Then we resolve the lower bound on the quantum communication complexity of the instance $\mathrm{ESetInc}_{2k,k,l,1/2}^{4k}$ (See \Cref{lemma:partial}). Compared to \cite{GKS16}, our reduction is more concise and direct.
 %by reducing to the approximate degree of a corresponding partial symmetric Boolean functions by pattern matrix method (See \Cref{lemma:partial}). 

For the communication complexity of $\SetInc_{a,b,c,g}^n$ and $\ESetInc_{a,b,c,g}^n$ (or $\GHD_{a,b,a+b-2c,2g}$ and $\eGHD_{a,b,a+b-2c,2g}$ equivalently), the comparison between \cite{GKS16} and our results are shown as \Cref{tab:comparison}, where $n_1,n_2$ are the smallest two numbers in $a-c,b-c,c,n-a-b+c$ and $n_1 \le n_2$. It is worth noting that the results in \cite{GKS16} mainly depend on $n_2/g$, while our results mainly depend on $\sqrt{n_1n_2}/g$, which is a key factor in characterizing the quantum communication complexity of the Set-Inclusion Problem.

\begin{table}[htbp]
\renewcommand{\arraystretch}{2.2}
\centering
\caption{The communication complexity of Set-Inclusion Problem}
\label{tab:comparison}
\scalebox{0.9}{
\begin{tabular}{cl}
\toprule
%& \textbf{Upper bound of $\SetInc_{a,b,c,g}^n$} & \textbf{Lower bound of $\ESetInc_{a,b,c,g}^n$} \\
& Communication complexity\\
\midrule
\multirow{2}{*}{Ref. \cite{GKS16}\tablefootnote{Ref. \cite{GKS16} gave the communication complexity of $\GHD_{a,b,c,g}$ and $\eGHD_{a,b,c,g}$ in Lemmas 3.3 and 3.4 actually. We reformulated their formulas for convenience. The details can be seen in {\sc Appendix} \ref{appendix_ghd}.}} &  Randomized: $O\inparen{\inparen{\frac{n_2}{g}}^2\log \frac{b}{g}}$ ($a \le b \le \frac{n}{2}$) \\
& Randomized: $\Omega\inparen{\max\inbrace{\frac{n_2}{g}, \log \inparen{\frac{\min\inbrace{a+b-2c,n-a-b+2c}}{g}}}}$  \\
\\
\multirow{3}{*}{\textbf{Our work}} & Randomized: $O\inparen{\frac{n_1n_2}{g^2}\log n}$ \\
& Quantum: $O\inparen{\frac{\sqrt{n_1n_2}}{g}\log n}$ \\
& Quantum: $\Omega\inparen{\frac{\sqrt{n_1n_2}}{g}}$ \\
\bottomrule
\end{tabular}
}
\end{table}

\begin{remark}
    It is worth noting that for both quantum and randomized communication complexities, the upper bound of $\SetInc_{a,b,c,g}^n$ is also the upper bound of $\ESetInc_{a,b,c,g}^n$. Similarly, the lower bound of $\ESetInc_{a,b,c,g}^n$ is also the lower bound of $\SetInc_{a,b,c,g}^n$. 
\end{remark}

\subsection{Organization}
The remaining part of the paper is organized as follows. In \Cref{sec:pre}, we state some notations and definitions used in this paper. In \Cref{sec:cc}, we study the quantum and classical communication complexities of permutation-invariant functions. In \Cref{sec:log_rank}, we show the Log-rank Conjecture holds for non-trivial total permutation-invariant functions.
  In \Cref{sec:log_arank}, we study the Log-approximate Conjecture of permutation-invariant functions both in quantum and classical setting.
Finally, a conclusion is made in \Cref{sec:con}. The appendices contain a section on extended preliminaries and omitted proofs.

%%%%%%%%%%%%%%%%%%%%%%%%%%%%%%%%%%%%%%%%%%%%%%%%%%%%%%%%%%%%%%%%%%%%%%%%%%%%%%%%
%%%%%%%%%%%%%%%%%%%%%%%%%%%%%%%%%%%%%%%%%%%%%%%%%%%%%%%%%%%%%%%%%%%%%%%%%%%%%%%%
%%%%%%%%%%%%%%%%%%%%%%%%%%%%%%%%%%%%%%%%%%%%%%%%%%%%%%%%%%%%%%%%%%%%%%%%%%%%%%%%
\section{Preliminaries}\label{sec:pre}
We introduce the notations and definitions used in this paper.

A \defemph{multiset} is a set with possibly repeating elements. We use $\{[\cdot]\}$ to denote multiset and $\{\cdot\}$ to denote standard set. Let $S$ be a multiset, $S\setminus \inbrace{a}$ removes one occurrence of $a$ from $S$ if there is any. 

%%%%%%%%%%%%%%%%%%%%%%%%%%%%%%%%%%%%%%%%%%%%%%%%%%%%%%%%%%%%%%%%%%%%%%%%%%%%%%%%
%%%%%%%%%%%%%%%%%%%%%%%%%%%%%%%%%%%%%%%%%%%%%%%%%%%%%%%%%%%%%%%%%%%%%%%%%%%%%%%%
\subsection{Boolean Functions}

A \defemph{partial function} is a function defined only on a subset of its domain $\domain$. Formally, given a partial Boolean function $f\colon \domain \rightarrow \inbrace{-1,1,*}$, $f(x)$ is \defemph{undefined} for $x\in \domain$ if $f(x)=*$. A \defemph{total function} is a function that is defined on the entire domain. We say $f\colon \domain\rightarrow \inbrace{-1,1,*}$ is a \defemph{subfunction} of $g\colon \domain\rightarrow \inbrace{-1,1,*}$ if 
$f(x) = g(x)$ or $f(x) = *$ for any $x \in \domain$. For function $f\colon \domain \rightarrow \inbrace{-1,1,*}$, we define $\overline{f}\colon \domain \rightarrow \inbrace{-1,1,*}$ as 
\begin{align*}
\overline{f}(x)\coloneqq
\begin{cases}
-f(x) &\text{if $f(x) = 1$ or $-1$}, \\
* &\text{otherwise.}
\end{cases}
\end{align*}

A \defemph{Boolean predicate} is a partial function that has domain $\domain = \inbrace{0,1,...,n}$ for any $n\in\mathbb{Z}^+$. 

An \defemph{incomplete Boolean matrix} is a matrix with entries in $\inbrace{-1,1,*}$, where undefined entries are filled with $*$. 

A \defemph{submatrix} is a matrix that is obtained by extracting certain rows and/or columns from a given matrix.

A \defemph{half-integer} is a number of the form $n+1/2$, where $n \in \mathbb{Z}$.

We introduce some Boolean operators as follows. For every $n\in\mathbb{N}$ and $x,y\in\inbrace{0,1}^n$:
\begin{enumerate}
    \item $\overline{x} \coloneqq (\overline{x_0},...,\overline{x_{n-1}}) =
    (1-x_0,...,1-x_{n-1})$;
    \item $x \land y \coloneqq (x_0 \land y_0,..., x_{n-1} \land y_{n-1})$; and
    \item $x \oplus y \coloneqq (x_0 \oplus y_0,..., x_{n-1} \oplus y_{n-1})$.
\end{enumerate}

%%%%%%%%%%%%%%%%%%%%%%%%%%%%%%%%%%%%%%%%%%%%%%%%%%%%%%%%%%%%%%%%%%%%%%%%%%%%%%%%
%%%%%%%%%%%%%%%%%%%%%%%%%%%%%%%%%%%%%%%%%%%%%%%%%%%%%%%%%%%%%%%%%%%%%%%%%%%%%%%%
\subsection{Communication Complexity Model}\label{sec:cc_model}
In the two-party communication model, Alice is given input $x \in \B^n$, and Bob is given input $y \in \B^n$. Then they aim to compute $f(x,y)$ for some function $f: \{0,1\}^n \times \{0,1\}^n \to \{-1,1,*\}$ by communication protocols while minimizing
the amount of communication between them.
%The \defemph{deterministic communication complexity $D(f)$} is defined as the cost of the deterministic protocol with the smallest communication cost, which computes $f$ correctly on any input. The \defemph{randomized communication complexity $R_{\epsilon}(f)$} is defined as the cost of the randomized protocol with the smallest communication cost, which has access to public randomness and computes $f$ correctly on any input with probability at least $1-\epsilon$. Similarly, the \defemph{quantum communication complexity $Q_{\epsilon}(f)$} is defined as the cost of the quantum protocol with the smallest cost, which is allowed to share prior entanglement and computes $f$ correctly on any input with probability at least $1-\epsilon$. %If the quantum protocol is allowed with prior entanglement initially, then the corresponding quantum communication complexity is denoted $Q(f)$. If a protocol succeeds with probability at least $1-\epsilon$ on any input for some constant $\epsilon < 1/2$, we say the protocol is with \defemph{bounded error}. If $\epsilon = 1/3$, we abbreviate $R_{\epsilon}(f),Q_{\epsilon}(f)$ as $R(f),Q(f)$. 
%Given protocol $\Pi$, let $\mathrm{CC}(\Pi)$ be defined as the number of . 
{In this paper, we consider the communication protocols in the deterministic, randomized, and quantum settings, respectively. Furthermore, the formal definitions of communication complexities are shown as follows.} 

\begin{definition}
	[Deterministic Communication Complexity] For any function $f: \bit^n\times \bit^n\rightarrow \{-1,1,*\}$, the deterministic communication complexity $D(f)$ is defined as the cost of the deterministic protocol with the smallest cost, which computes $f$ correctly on any input. In particular, 
	$$D(f)=\min_{\substack{\Pi: \text{ deterministic protocol}\\ \forall x,y\in\B^n \text{ s.t. }  \Pi(x,y)=f(x,y)}}\CC\inparen{\Pi},$$
	where $\CC\inparen{\Pi}$ is the number of bits exchanged between two parties in the protocol $\Pi$.
\end{definition}

\begin{definition}
	[Randomized Communication Complexity] For any function $f: \bit^n\times \bit^n\rightarrow \{-1,1,*\}$, the randomized communication complexity $R_{\epsilon}(f)$ is defined as the cost of the randomized protocol with the smallest cost, which has access to public randomness and computes f correctly on any input with probability at least $1-\epsilon$. In particular, 
	$$R_{\epsilon}(f)=\min_{\substack{\Pi: \text{ randomized protocol}\\ \forall x,y\in\B^n \text{ s.t. } f(x,y)\neq *: \Pr[\Pi(x,y)]=f(x,y)]\ge 1-\epsilon}}\CC\inparen{\Pi},$$
	where $\CC\inparen{\Pi}$ is the number of bits exchanged between two parties in the protocol $\Pi$.
\end{definition}

\begin{definition}
	[Quantum Communication Complexity] For any function $f: \bit^n\times \bit^n\rightarrow \{-1,1,*\}$, the quantum communication complexity $Q_{\epsilon}(f)$ is defined as the cost of the quantum protocol with the smallest cost, which is allowed to share prior entanglement and computes $f$ correctly on any input with probability at least $1-\epsilon$. In particular, 
	$$Q_{\epsilon}(f)=\min_{\substack{\Pi: \text{ quantum protocol}\\ \forall x,y\in\B^n \text{ s.t. } f(x,y)\neq *: \Pr[\Pi(x,y)]=f(x,y)]\ge 1-\epsilon}}\CC\inparen{\Pi},$$
		where $\CC\inparen{\Pi}$ is the number of qubits exchanged between two parties in the protocol $\Pi$. %In the quantum protocols, two parties can share prior entanglement.
\end{definition}

If a protocol succeeds with probability at least $1-\epsilon$ on any input for some constant $\epsilon < 1/2$, we say the protocol is with \defemph{bounded error}. If $\epsilon = 1/3$, we abbreviate $R_{\epsilon}(f),Q_{\epsilon}(f)$ as $R(f),Q(f)$. It is a folklore conclusion that $Q(f) \le R(f) \le D(f)$.

%%%%%%%%%%%%%%%%%%%%%%%%%%%%%%%%%%%%%%%%%%%%%%%%%%%%%%%%%%%%%%%%%%%%%%%%%%%%%%%%
%%%%%%%%%%%%%%%%%%%%%%%%%%%%%%%%%%%%%%%%%%%%%%%%%%%%%%%%%%%%%%%%%%%%%%%%%%%%%%%%

\subsection{Permutation-Invariant Functions}

{By \Cref{remark:and},} any permutation-invariant function $f:\inbrace{0,1}^n\times \inbrace{0,1}^n \rightarrow \inbrace{-1,1,*}$ depends only on $|x|,|y|$ and $|x \land y|$. Thus, for any $a,b \in \inbrace{0,1,...,n}$, there exists a function $f_{a,b}\colon\inbrace{\max\inbrace{0,a+b-n},...,\min\inbrace{a,b}} \rightarrow \inbrace{-1,1{,*}}$ 
such that
\begin{equation}\label{eq:hab}
    f_{a,b}(|x\land y|) = f(x,y),
\end{equation}
for any $x,y \in \inbrace{0,1}^n$ satisfying $|x| = a, |y| = b$.
If there exist $a,b \in [n]$ such that $f_{a,b}$ is not a constant function, we say $f$ is \defemph{non-trivial}. 

The following definition of jumps partitions the domain of $f_{a,b}$ into different intervals according to the transition of function values.
\begin{definition}[Jump in $f_{a,b}$]\label{def:jump} 
{For $f_{a,b}$ defined as (\ref{eq:hab}), consider $c$ and $g$ such that $c+g,c-g$  are in the domain of $f_{a,b}$, we say $(c,g)$ is a jump in $f_{a,b}$ if all the following three conditions are satisfied:}
\begin{enumerate}
    \item $f_{a,b}(c-g) \neq f_{a,b}(c+g)$;
    \item $f_{a,b}(c-g),f_{a,b}(c+g) \in \{-1,1\}$;
    \item $f_{a,b}(r)$ is undefined for $c-g < r < c+g$.
\end{enumerate}
Moreover, we define $\mathscr{J}(f_{a,b})$ to be the set of all jumps in $f_{a,b}$:
$$
\mathscr{J}\left(f_{a, b}\right) \coloneqq\left\{(c, g): \begin{array}{c}
f_{a, b}(c-g), f_{a, b}(c+g) \in\{0,1\} \\
f_{a, b}(c-g) \neq f_{a, b}(c+g) \\
\forall i \in(c-g, c+g), f_{a, b}(i)=*
\end{array}\right\}.
$$
\end{definition}
\begin{remark}
    %The definition of jumps are different from \cite{GKS16}. Specifically, in our definition, 
    If $(c,g)$ is a jump in $f_{a,b}$, then one of $\ESetInc_{a,b,c,g}^n$ and $\overline{\ESetInc}_{a,b,c,g}^n$ is a subfunction of $f$.  Additionally, Definition \ref{def:jump} is adapted from \cite{GKS16}, while \cite{GKS16} gives the definition of jumps in $h_{a,b}$, where $h_{a,b}(\Delta(x,y)) = f(x,y)$ for permutaion-invariant function $f$ and $\Delta(x,y)$ is the Hamming distance between $x$ and $y$.
\end{remark}

%The following definition gives an important instance of permutation-invariant functions. 

The following measure $m(\cdot)$ is used to capture the quantum communication complexity of permutation-invariant functions, which is inspired by the complexity measure introduced in~\cite{GKS16}.

\begin{definition}[Measure $m(f)$]
\label{def:mf}
Fix $n\in\mathbb{Z^+}$. Let $f:\{0,1\}^n \times \{0,1\}^n \to \{-1,1,*\}$ be a permutation-invariant function. If $f$ is non-trivial, we define the measure $m(f)$ of $f$ as follows:
\[
m(f) \coloneqq
\max_{\substack{a, b \in \inbrace{1,2,...,n-1} \\(c, g) \in \mathscr{J}\left(f_{a, b}\right) \\ n_1\coloneqq \min\{[a-c,c,b-c,n-a-b+c]\}\\n_2\coloneqq \min\left(\{[a-c,c,b-c,n-a-b+c]\}\setminus \inbrace{n_1}\right) }} \frac{\sqrt{n_{1} n_{2}}}{g}.
\]
{If $f(x,y)$ only depends on $|x|$ and $|y|$, let $m(f)\coloneqq 0$.}
\end{definition}
\begin{remark}
\Cref{def:mf} is motivated by the quantum lower bound of the Exact Set-Inclusion Problem (See \Cref{th:qcc_lowbound}). That is, $m(f)$ is the maximum of the quantum lower bound of all subfunctions of $f$ that are isomorphic to an Exact Set-Inclusion Problem. As a comparison, a similar measure in \cite{GKS16} is defined as the maximum of the randomized lower bound of all subfunctions of $f$ that are isomorphic to an Exact Gap-Hamming-Distance Problem.
\end{remark}

%%%%%%%%%%%%%%%%%%%%%%%%%%%%%%%%%%%%%%%%%%%%%%%%%%%%%%%%%%%%%%%%%%%%%%%%%%%%%%%%
%%%%%%%%%%%%%%%%%%%%%%%%%%%%%%%%%%%%%%%%%%%%%%%%%%%%%%%%%%%%%%%%%%%%%%%%%%%%%%%%
%%%%%%%%%%%%%%%%%%%%%%%%%%%%%%%%%%%%%%%%%%%%%%%%%%%%%%%%%%%%%%%%%%%%%%%%%%%%%%%%
\section{Polynomial Equivalence on Communication Complexity of Permutation-Invariant Functions}\label{sec:cc}

To show the polynomial equivalence between quantum and randomized communication complexity of permutation-invariant functions as stated in \Cref{th:main}, we prove
the following two theorems (proved in \Cref{subsec:qcc_lowerbound,subsec:upper-bound}, respectively) for the quantum and randomized communication complexities of permutation-invariant functions using the measure in \Cref{def:mf}.

\begin{theorem}[Lower Bound]
\label{th:low}
Fix $n\in{\mathbb{Z}^+}$. Let $f:\{0,1\}^n \times \{0,1\}^n \to \{-1,1,*\}$ be a permutation-invariant function. We have 
\[
Q(f) = \Omega(m(f)).
\]
\end{theorem}

\begin{theorem}[Upper Bound]
\label{th:upper}
Fix $n\in{\mathbb{Z}^+}$.
Given a permutation-invariant function $f:\{0,1\}^n \times \{0,1\}^n \to \{-1,1,*\}$ and the corresponding measure $m(f)$ defined in Definition \ref{def:mf}, we have
\begin{enumerate}
    \item $R(f) = O\inparen{m(f)^2 \log^2 n\log\log n+\log n}$, and
    \item $Q(f) = O\inparen{m(f) \log^2 n\log\log n+\log n}$.
\end{enumerate}
\end{theorem}

%%%%%%%%%%%%%%%%%%%%%%%%%%%%%%%%%%%%%%%%%%%%%%%%%%%%%%%%%%%%%%%%%%%%%%%%%%%%%%%%
%%%%%%%%%%%%%%%%%%%%%%%%%%%%%%%%%%%%%%%%%%%%%%%%%%%%%%%%%%%%%%%%%%%%%%%%%%%%%%%%
\subsection{Quantum Communication Complexity Lower Bound}\label{subsec:qcc_lowerbound}
In this section, our goal is to obtain a lower bound on the quantum communication complexity for permutation-invariant functions (\Cref{th:low}). Towards this end, we show that every permutation-invariant function $f$ can be reduced to $\mathrm{ESetInc}$ (defined in \Cref{def:setinc}) and exhibit a lower bound for $\mathrm{ESetInc}$ (\Cref{th:qcc_lowbound}). Additionally, \Cref{th:qcc_lowbound} implies if $|x|=a,|y|=b$, then the cost to distinguish $|x\land y| =c-g$ from $|x\land y| =c+g$ is related to the smallest two numbers in $\{[a-c,c,b-c,n-a-b+c]\}$.

\begin{lemma}
\label{th:qcc_lowbound}
Fix $n\in{\mathbb{Z^+}}$. Consider $a,b \in \inbrace{1,...,n-1}$ and $c-g,c+g$ are achievable Hamming weights of $|x \land y|$ when $|x| = a, |y| = b$. Let $n_1\coloneqq \min\{[a-c,c,b-c,n-a-b+c]\}$ and $n_2\coloneqq \min\left(\{[a-c,c,b-c,n-a-b+c]\}\setminus \inbrace{n_1}\right)$. We have
\[Q\inparen{\mathrm{ESetInc}_{a,b,c,g}^n} = \Omega\inparen{\frac{\sqrt{n_1 n_2}}{g}}.\] 
\end{lemma}

\begin{proof}[Proof of \Cref{th:low}]
By the definitions of $f_{a,b}$ and jump of $f_{a,b}$,
any quantum protocol computing $f$ can also compute $\mathrm{ESetInc}^{n}_{a,b,c,g}$ for any $a,b$ and any jump $(c,g) \in \mathscr{J}(f_{a,b})$. Therefore, given a jump $(c,g)$ for $f_{a,b}$, the cost of computing $\mathrm{ESetInc}^{n}_{a,b,c,g}$ lower bounds the cost of computing $f$. By \Cref{th:qcc_lowbound}, we have $Q(f) \ge \Omega\inparen{\frac{\sqrt{n_1n_2}}{g}}$ for any jump $(c,g)$ in $f_{a,b}$, where $n_1, n_2$ are the smallest two numbers in $\{[a-c,c,b-c,n-a-b+c]\}$.
We conclude that
$Q(f) = \Omega\inparen{m(f)}$ as desired.
\end{proof}

Now we remain to show \Cref{th:qcc_lowbound}. We note that the following two lemmas imply \Cref{th:qcc_lowbound} directly, where \Cref{le:qcc_lowbound1} reduces the instance such that the parameter only relies on $n_1,n_2,g$ and \Cref{le:qcc_lowbound2} gives the final lower bound.  

\begin{restatable}{lemma}{qcclowerboundone}
\label{le:qcc_lowbound1}
Fix $n\in{\mathbb{Z^+}}$. Consider $a,b \in \inbrace{1,...,n-1}$ and $c-g,c+g$ are achievable Hamming weights of $|x \land y|$ when $|x| = a, |y| = b$. 
Let $n_1\coloneqq \min\{[a-c,c,b-c,n-a-b+c]\}$ and $n_2\coloneqq \min\left(\{[a-c,c,b-c,n-a-b+c]\}\setminus \inbrace{n_1}\right)$. We have
\[Q\inparen{\mathrm{ESetInc}_{a,b,c,g}^n} \ge Q\inparen{\mathrm{ESetInc}_{n_1+n_2,n_1+n_2,n_1,g}^{n_1+3n_2}}.\] 
\end{restatable}

\begin{lemma}\label{le:qcc_lowbound2}
Consider $n_1,n_2,g$ such that  $n_1\le n_2$ and $n_1-g,n_1+g$ are achievable Hamming weights of $|x \land y|$ when $|x| = |y| = n_1+n_2$, we have
\[Q\inparen{\mathrm{ESetInc}_{n_1+n_2,n_1+n_2,n_1,g}^{n_1+3n_2}} = \Omega\inparen{\frac{\sqrt{n_1 n_2}}{g}}.\]
\end{lemma}

We use the following two results on $\ESetInc$ to show \Cref{le:qcc_lowbound1,le:qcc_lowbound2}. Specifically, \Cref{le:SecInc} is a variant of Lemma~4.1 in \cite{CR12} and shows some reduction methods to the instances of the Exact Set-Inclusion Problem. \Cref{lemma:partial} is a generalization of Theorem~5 in \cite{BCG+21} proved by pattern matrix method and shows the lower bound of a special instance of the Exact Set-Inclusion Problem. The proofs of \Cref{le:SecInc,lemma:partial} are given in Appendix \ref{appendix_qcc_lowerbound}.

\begin{restatable}{lemma}{QSecInc}
\label{le:SecInc}
Fix $n\in{\mathbb{Z^+}}$. Consider $a,b \in \inbrace{1,...,n-1}$ and $c-g,c+g$ are achievable Hamming weights of $|x \land y|$ when $|x| = a, |y| = b$.
The following relations hold.
\begin{enumerate}
    \item $Q\inparen{\mathrm{ESetInc}_{a, b, c, g}^{n}} \le Q\inparen{\mathrm{ESetInc}_{a+\ell_1 + \ell_3, b + \ell_2 + \ell_3, c+\ell_3, g}^{n+\ell}}$ for integers $\ell_1, \ell_2, \ell_3 \ge 0$ such that $\ell_1 + \ell_2 + \ell_3 \le \ell$;
    \item $Q\inparen{\mathrm{ESetInc}_{a, b, c, g}^{n}} = Q\inparen{\mathrm{ESetInc}_{a, n-b, a-c, g}^{n}} = Q\inparen{\mathrm{ESetInc}_{n-a, b, b-c, g}^{n}}$;
    \item $Q\inparen{\mathrm{ESetInc}_{a, b, c, g}^{n}} \le Q\inparen{\mathrm{ESetInc}_{ka, kb, kc, kg}^{kn}}$, where $k \ge 1$ is an integer.
\end{enumerate}
\end{restatable}

\begin{restatable}{lemma}{lemmapartial}
\label{lemma:partial}
For every $k\in{\mathbb{Z}^+}$,
if $l$ is a half-integer and $0 < l \le k/2$, then    $Q\inparen{\mathrm{ESetInc}^{4k}_{2k,k,l,1/2}} = \Omega\inparen{\sqrt{kl}}$.
\end{restatable}

\begin{proof}[Proof of \Cref{le:qcc_lowbound1}]
Using the second item of \Cref{le:SecInc}, we assume $n_1 = c$ without loss of generality. Furthermore, we assume $n_2 = a-c$. Let $n_3 \coloneqq b-c, n_4 \coloneqq n-a-b+c$. Then $n_3,n_4 \ge n_2 \ge n_1$ and $n = n_1+n_2+n_3+n_4$.
    By Lemma \ref{le:SecInc}, we have \[
    \begin{aligned}
&Q\inparen{\mathrm{ESetInc}^{n_1 + 3n_2}_{n_1 + n_2, n_1 + n_2, n_1, g}}\\
&= Q\inparen{\mathrm{ESetInc}^{n_1+n_2+n_2+n_2}_{n_1+n_2,n_1+n_2, n_1, g}} \\
&\le Q\inparen{\mathrm{ESetInc}^{n_1+n_2+n_3+n_4}_{n_1+n_2, n_1+n_3, n_1, g}} \\
&= Q\inparen{\mathrm{ESetInc}^n_{a,b,c,g}}.
    \end{aligned}
\]
If $n_2 = b-c$ or $n-a-b+c$, the argument is similar.
\end{proof}

\begin{proof}[Proof of \Cref{le:qcc_lowbound2}]
Let $m_1 = \left\lfloor\frac{n_1}{2g}+\frac{1}{2}\right\rfloor-\frac{1}{2}$, i.e., $m_1$ is the largest half-integer no more than $\frac{n_1}{2g}$.
Similarly, let $m_2 = \left\lfloor\frac{n_2}{2g}+\frac{1}{2}\right\rfloor-\frac{1}{2}$. By  \Cref{le:SecInc}, we have
\[
\begin{aligned}
&Q\inparen{\mathrm{ESetInc}^{n_1+3n_2}_{n_1+n_2, n_1+n_2, n_1, g}} \\
&\ge Q\inparen{\mathrm{ESetInc}^{m_1+3m_2}_{m_1+m_2, m_1+m_2, m_1, 1/2}}.
\end{aligned}
\]
Then we discuss the following three cases:
\begin{itemize}
    \item \emph{Case 1: $m_1 = m_2 = 1/2$.}
We have 
\[
\begin{aligned}
&Q\inparen{\mathrm{ESetInc}^{m_1+3m_2}_{m_1+m_2, m_1+m_2, m_1, 1/2}} \\
&= 
\Omega\inparen{1} \\
&= \Omega\inparen{\sqrt{m_1m_2}}.
\end{aligned}
\]
\item \emph{Case 2: $m_2 \ge 3/2$ and $m_1 = 1/2$.}
 Let 
$m'_2 \coloneqq \left\lfloor \frac{m_1+m_2}{2}\right\rfloor, 
l_1 \coloneqq m_1+m_2-2m'_2,
l_2 \coloneqq m_1+m_2-m'_2,
l \coloneqq m_1+3m_2-4m'_2$.
Then, 
\[
\begin{aligned}
l-(l_1+l_2) = m_2-m_1-m'_2
\ge \frac{m_2+m_1}{2}-m'_2 
\ge 0.
\end{aligned}
\]
By \Cref{{le:SecInc},lemma:partial}, we have
\[
\begin{aligned}
&Q\inparen{\mathrm{ESetInc}^{m_1+3m_2}_{m_1+m_2, m_1+m_2, m_1, 1/2}}\\
&=
Q\inparen{\mathrm{ESetInc}^{4m'_2+l}_{2m'_2+l_1, m'_2+l_2, m_1, 1/2}}\\
&\ge 
Q\inparen{\mathrm{ESetInc}^{4m'_2}_{2m'_2, m'_2, m_1, 1/2}}\\
&= \Omega\inparen{\sqrt{m_1m'_2}}\\
&= \Omega\inparen{\sqrt{m_1m_2}}.
\end{aligned}
\]

\item \emph{Case 3: $m_1 \ge 3/2$.}
Let $
m \coloneqq \left\lfloor \frac{m_1}{6}+\frac{m_2}{2}\right\rfloor,
k \coloneqq \left\lfloor\frac{m_1}{3}+\frac{1}{2}\right\rfloor-\frac{1}{2},
l_3 \coloneqq m_1-k,
l_1 \coloneqq \inparen{m_1+m_2-2m}-l_3,
l_2 \coloneqq \inparen{m_1+m_2-m}-l_3,
l \coloneqq m_1+3m_2-4m$.
Since $k$ is the largest half-integer smaller than $\frac{m_1}{3}$, we have $k \le \frac{1}{2}\cdot \left\lfloor \frac{2m_1}{3}\right\rfloor$.
Since $m_1 \le m_2$, we have 
\begin{equation}\label{eq:m1m}
k \le  \frac{1}{2}\cdot\left\lfloor \frac{2m_1}{3}\right\rfloor \le \frac{1}{2} \cdot\left\lfloor \frac{m_1}{6}+\frac{m_2}{2}\right\rfloor \le \frac{m}{2},
\end{equation}
and
\begin{equation}\label{eq:ll_com}
\begin{aligned}
l-\inparen{l_1+l_2+l_3} &= m_2-k-m \\
&\ge m_2-\frac{m_1}{3}-\inparen{\frac{m_1}{6}+\frac{m_2}{2}} \\
&\ge 0.
\end{aligned}
\end{equation}
Then we have
\[
\begin{aligned}
&Q\inparen{\mathrm{ESetInc}^{m_1+3m_2}_{m_1+m_2, m_1+m_2, m_1, 1/2}} \\
&= 
Q\inparen{\mathrm{ESetInc}^{4m+l}_{2m+l_1+l_3, m+l_2+l_3, k+l_3, 1/2}} \\
&\ge 
Q\inparen{\mathrm{ESetInc}^{4m}_{2m, m, k, 1/2}} \\% & &\inparen{\text{by  \Cref{le:SecInc} and (\ref{eq:ll_com})}}\\
&= \Omega\inparen{\sqrt{mk}} \\ %& &\inparen{\text{by   \Cref{lemma:partial} and (\ref{eq:m1m})}}\\
&= 
\Omega\inparen{\sqrt{m_1m_2}},
\end{aligned}
\]
where the inequality follows from  \Cref{le:SecInc} and (\ref{eq:ll_com}) and the second equality follows from \Cref{lemma:partial} and (\ref{eq:m1m}).
\end{itemize}
We conclude that
\[
\begin{aligned}
&Q\inparen{\mathrm{ESetInc}^{n_1+3n_2}_{n_1+n_2, n_1+n_2, n_1, g}} \\
&\ge Q\inparen{\mathrm{ESetInc}^{m_1+3m_2}_{m_1+m_2, m_1+m_2, m_1, 1/2}}\\
&= \Omega\inparen{\frac{\sqrt{n_1n_2}}{g}}.
\end{aligned}
\]
\end{proof}

%%%%%%%%%%%%%%%%%%%%%%%%%%%%%%%%%%%%%%%%%%%%%%%%%%%%%%%%%%%%%%%%%%%%%%%%%%%%%%%%
%%%%%%%%%%%%%%%%%%%%%%%%%%%%%%%%%%%%%%%%%%%%%%%%%%%%%%%%%%%%%%%%%%%%%%%%%%%%%%%%
\subsection{Randomized and Quantum Communication Complexity Upper Bound}
\label{subsec:upper-bound}

We show upper bounds on the randomized and quantum communication complexities for permutation invariant functions (\Cref{th:upper}). Similar to \Cref{subsec:qcc_lowerbound}, we do so by giving upper bounds for a specific problem, $\SetInc$ (see \Cref{def:setinc}), and reducing permutation-invariant functions to $\SetInc$. %The intuition to design protocols of the Set-Inclusion Problem is as follows.

%\begin{remark}
	The intuition of our randomized protocol to compute $\SetInc$ is as follows: Let $s_1,s_2$ $(s_1 \le s_2)$ be the smallest two numbers in $\{[|\overline{x} \land y|, |x \land y|, |\overline{x} \land y|, |\overline{x} \land \overline{y}|]\}$, and $n_1,n_2$ $(n_1 \le n_2)$ be the smallest two numbers in $\{[a-c,c,b-c,n-a-b+c]\}$. It is worth noting that $|n_1-s_1| = |n_2-s_2| = g$. Let $p = \frac{s_1}{s_1+s_2}$. Then we show $\SetInc$ is equivalent to distinguish $p \le \beta-\epsilon$ from $p \ge\beta+\epsilon$ for some $\beta = O\inparen{\frac{n_1}{n_2}}, \epsilon = \Omega\inparen{\frac{n_2}{g}}$.
	Depending on the value of $n_1,n_2$, we use different sampling methods to estimate $p$ with error $\epsilon$. In the bounded-error randomized case, the communication cost is $O\inparen{\frac{\beta}{\epsilon^2}\log n} = O\inparen{\frac{n_1n_2}{g^2}\log n}$. Then we repeat $O(\log \log n)$ times of protocol such that the failed probability is at most $O(1/\log n)$. Finally, we use quantum amplitude amplification to speed up the randomized protocol.
%\end{remark}

The following two lemmas capture the randomized and quantum communication complexity for $\SetInc$, respectively.

\begin{lemma}[Randomized Upper Bound]
\label{le:extiXandY}
Fix $n,a,b\in{\mathbb{Z}^+}$. Consider $c,g$ such that $c+g,c-g\in{\mathbb{N}}$. Let $n_1\coloneqq \min\{[a-c,c,b-c,n-a-b+c]\}$ and $n_2\coloneqq \min\left(\{[a-c,c,b-c,n-a-b+c]\}\setminus \inbrace{n_1}\right)$. For any input $x,y\in \B^n$ of $\mathrm{SetInc}_{a, b, c, g}^{n}$, there exists a randomized communication protocol that computes $\mathrm{SetInc}_{a, b, c, g}^{n}(x,y)$ using $O\inparen{\frac{n_1n_2}{g^2} \log n \log \log n}$  bits of communication with success probability at least $1- 1/(6\log n)$.
\end{lemma}

\begin{restatable}[Quantum Upper Bound]{lemma}{QuanUpperSetInc}
\label{lemma:quantum_upper_bound_setinc}
Fix $n,a,b\in{\mathbb{Z}^+}$. Consider $c,g$ such that $c+g,c-g\in{\mathbb{N}}$. Let $n_1\coloneqq \min\{[a-c,c,b-c,n-a-b+c]\}$ and $n_2\coloneqq \min\left(\{[a-c,c,b-c,n-a-b+c]\}\setminus \inbrace{n_1}\right)$. For any input $x,y\in \B^n$ of $\mathrm{SetInc}_{a, b, c, g}^{n}$, there exists a quantum communication protocol without prior entanglement that computes $\mathrm{SetInc}_{a, b, c, g}^{n}(x,y)$ using $O\inparen{\frac{\sqrt{n_1n_2}}{g} \log n \log \log n}$ qubits of communication with success probability at least  $1- 1/(6\log n)$.
\end{restatable}
%We note that \Cref{lemma:quantum_upper_bound_setinc} is a quantum speedup version of \Cref{le:extiXandY} by quantum amplitude amplification. 
The proof of \Cref{le:extiXandY} is given at the end of this section, and we postpone the proof of \Cref{lemma:quantum_upper_bound_setinc} to \Cref{appendix_quantum_speedup}.

Now we explain how to derive \Cref{th:upper} from the lemmas above.

\begin{proof}[Proof of \Cref{th:upper}]
We first present a randomized protocol to compute $f$ based on binary search:
\begin{enumerate}
    \item Alice and Bob exchange $a\coloneqq |x|, b\coloneqq |y|$.
    \item Alice and Bob both derive $f_{a,b}$ such that $f_{a,b}(|x \land y|) = f(x,y)$.
    \item Let $\mathscr{J}\left(f_{a, b}\right)=\{(c_i,g_i)\}_{i\in[k]}$ for some $k\leq n$ be the set of jumps of $f_{a,b}$ as in \Cref{def:jump}.
    \item Alice and Bob use binary search to determine $i\in \{0,1,...,k\}$ such that $\abs{x\land y}\in I_i$, where $I_i$ is defined in (\ref{eq:Ik}).
\end{enumerate}

We first discuss the communication complexity of the above protocol. The first step takes $O(\log n)$ bits of communication. The fourth step costs $O\inparen{m(f)^2 \log^2 n\log\log n}$ bits of communication: For each $i\in[k]$, Alice and Bob can determine whether $|x \land y| \le c_i-g_i$ or $|x \land y| \ge c_i+g_i$ by $O(m(f)^2 \log n\log\log n)$ communication cost with a success probability of at least $1-1/ \inparen{6\log n}$ by \Cref{le:extiXandY}. Since binary search takes at most $\lceil \log \inparen{k+1} \rceil = O\inparen{\log n}$ rounds, the total communication cost is $O\inparen{m(f)^2 \log^2 n\log\log n+\log n}$.

Now we argue for the correctness of the protocol. Notice that the set of jumps $\mathscr{J}\left(f_{a, b}\right)$ invokes $k+1$ intervals:
\begin{equation}\label{eq:Ik}
\begin{aligned}
I_0&\coloneqq [0,c_1-g_1],\\
I_1 &\coloneqq [c_1+g_1,c_2-g_2],\\
&\dots\\
I_{k-1} &\coloneqq [c_{k-1}+g_{k-1},c_k-g_k],\\
I_k &\coloneqq [c_k+g_k,n].
\end{aligned}
\end{equation}
\noindent In particular, the followings hold:
\begin{itemize}
    \item For every $j\in[0,k]$ and $z_1,z_2\in I_j$ such that $f_{a,b}(z_1)\neq *$ and $f_{a,b}(z_2)\neq *$, we have $f_{a,b}(z_1)=f_{a,b}(z_2)$.
    \item If $z\notin I_j$ for any $j\in[0,k]$, then $f_{a,b}(z)=*$.
\end{itemize}
Therefore, Alice and Bob start from $i = \lfloor (k+1)/2 \rfloor$ to determine whether $|x \land y| \le c_i-g_i$ or $|x \land y| \ge c_i+g_i$ with success probability of at least $1-1/ \inparen{6\log n}$. Depending on the result, they repeat the same process similar to binary search to find the interval that $|x \land y|$ falls in. After at most $\lceil \log \inparen{k+1} \rceil = O\inparen{\log n}$ repetitions, there is only one remaining interval and they can determine $f_{a,b}(|x \land y|)$. For $n \ge 2$, the failure probability of the above protocol is at most
\[
\begin{aligned}
1-\inparen{1-\frac{1}{6\log n}}^{\lceil \log \inparen{k+1} \rceil}
&\le \frac{\lceil \log \inparen{k+1} \rceil}{6 \log n} \\
&\le \frac{\lceil \log \inparen{n+1} \rceil}{3 \log n^2} \\
&\le \frac{1}{3}.
\end{aligned}
\]

For the quantum case, Alice and Bob use the same protocol above, but we invoke \Cref{lemma:quantum_upper_bound_setinc} to analyze the communication complexity. 

\end{proof}

\begin{proof}[Proof of \Cref{le:extiXandY}]
We rely on the following two claims to prove the lemma. 
\begin{fact}[{\cite[Lemma~30]{AA14}}]
\label{le:sample}
Fix $0 < \epsilon < \beta < 1$ such that $\beta+\epsilon \le 1$.
For a set $S$, suppose there is a subset $S'$ of $S$ such that $\frac{|S'|}{|S|} \le \beta-\epsilon$ or $\frac{|S'|}{|S|} \ge \beta+\epsilon$. Suppose we can sample from $S$ uniformly and ask whether the sample is in $S'$. Then we can decide whether $\frac{|S'|}{|S|} \le \beta-\epsilon$ or $\frac{|S'|}{|S|} \ge \beta+\epsilon$ by $O(\beta/\epsilon^2)$ samples, with success probability at least $2/3$.
\end{fact}

\begin{restatable}{fact}{unisample}
\label{fact:sample_oplus}
Suppose $x,y \in \{0,1\}^n$ are the inputs of Alice and Bob such that $|x| \neq |y|$. Alice and Bob can sample an element from $S \coloneqq \inbrace{i:x_i \neq y_i}$ uniformly using $O(\log n)$ bits of communication.
\end{restatable}
We refer interesting readers to \Cref{appendix_cc} for the proof of \Cref{fact:sample_oplus}. Now we prove the lemma by casing on the values of $n_1$ and $n_2$.
\begin{itemize}
\item \emph{Case 1: $n_1 = c$ and $ n_2 = a-c$.}
 According to \Cref{def:setinc}, we have either $\frac{|x \land y|}{|x|}\le \frac{c-g}{a}$  or $\frac{|x \land y|}{|x|} \ge \frac{c+g}{a}$. Alice and Bob estimate $\frac{|x \land y|}{|x|}$ as follows: Alice chooses an index $i$ such that $x_i = 1$ uniformly at random. Then Alice sends $i$ to Bob, and Bob checks whether $y_i = 1$. By \Cref{le:sample}, setting $\beta \coloneqq \frac{c}{a}$, $\epsilon \coloneqq \frac{g}{a}$, 
 Alice and Bob can decide whether $\frac{|x \land y|}{|x|}\le \frac{c-g}{a}$  or $\frac{|x \land y|}{|x|} \ge \frac{c+g}{a}$  with bounded error 
 using $O\inparen{\frac{ac}{g^2}} = O\inparen{\frac{n_1n_2}{g^2}} $ samples. Since $|x| = a$, using $O\inparen{\frac{n_1n_2}{g^2} \log \log n}$ samples, they can decide whether $|x \land y | \le c-g$  or $|x \land y | \ge c+g$  with success probability at least $1- 1/(6\log n)$ by error reduction. Thus, the communication complexity is $O\inparen{\frac{n_1n_2}{g^2} \log n\log \log n}$. 

\item \emph{Case 2: $n_1= a-c$ and $ n_2 = c$, or $n_1= a-c$ and $ n_2 = c$, or $n_1= c$ and $ n_2 =b-c$.} A similar argument as in Case 1 applies.

\item \emph{Case 3: $n_1 = c$ and $ n_2 = n-a-b+c$.} Since $n_1 \le n_2$, we have $a+b \le n$. 
Then we consider the following two cases:
\begin{enumerate}
    \item \label{a_b_less_n} \emph{Case 3.1: $a+b < n$.}
     Let $m \coloneqq n_1+n_2$, $p \coloneqq \frac{|x \land y|}{|\overline{x} \oplus y|}$.
Since 
\[
\begin{aligned}
|\overline{x} \oplus y| &= |x \land y| + |\overline{x} \land \overline{y}| \\
&= |x \land y|+ \inparen{n-\inparen{a+b-|x \land y|}} \\
&= n-(a+b)+2|x \land y|,
\end{aligned}
\]
we have
\[
p=
\frac{|x \land y|}{n-(a+b)+2|x\land y|} = \frac{1}{\frac{n-(a+b)}{|x \land y|}+2}.
\]
Notice that $p$ is an increasing function of $|x\land y|$. As a result, if $|x \land y| \le c-g$, then $p\le \frac{c-g}{m-2g}$; if $|x \land y| \ge c+g$, then $p\ge \frac{c+g}{m+2g}$. 
Let 
\begin{equation*}
\begin{aligned}
    \beta & \coloneqq  \frac{1}{2}\inparen{\frac{c+g}{m+2g}+ \frac{c-g}{m-2g}} \\
    &= \frac{cm-2g^2}{m^2-4g^2} \\
    &= O\inparen{\frac{c}{m}},\\
\end{aligned}
\end{equation*}
\begin{equation*}
    \begin{aligned}
\epsilon & \coloneqq \frac{1}{2}\inparen{\frac{c+g}{m+2g}- \frac{c-g}{m-2g}}  \\
&= \frac{gm}{m^2-4g^2}\\
&= \Omega\inparen{\frac{g}{m}}.
\end{aligned}
\end{equation*}
For any $x \in \B^n$, we let $S_x \coloneqq \inbrace{i:x_i = 1}$.
By \Cref{fact:sample_oplus}, Alice and Bob can sample $i$ from $S_{\overline{x} \oplus y}$ uniformly using $O(\log n)$ bits communication. 
Since $i \in S_{\overline{x} \oplus y}$, if $x_i = y_i = 1$, then $i \in S_{x \land y}$; if $x_i = y_i = 0$, then $i \notin S_{x\land y}$.
By \Cref{le:sample}, using $O\inparen{\frac{\beta}{\epsilon^2}} = O\inparen{\frac{mc}{g^2}} = O\inparen{\frac{n_1n_2}{g^2}}$ samples, Alice and Bob can decide whether $p \le \beta-\epsilon$ or $p \ge \beta + \epsilon$ with bounded error. Equivalently, Alice and Bob can distinguish $|x \land y | \le c-g$ from $|x \land y | \ge c+g$ with bounded error. By error reduction, using $O\inparen{\frac{n_1n_2}{g^2} \log \log n }$ samples, they can decide whether $|x \land y | \ge c-g$  or $|x \land y | \le c+g$  with success probability at least  $1- 1/(6\log n)$. Thus, the communication complexity is $O\inparen{\frac{n_1n_2}{g^2} \log n\log \log n}$. 

\item \emph{Case 3.2: $a+b = n$.} Alice and Bob generate new inputs $x' = x0$ and $y' = y0$ (pad a zero after the original input). Then we have 
\[
\mathrm{SetInc}_{a, b, c, g}^{n}(x,y) = \mathrm{SetInc}_{a, b, c, g}^{n+1}(x',y').
\]
Since $a+b < n+1$, Alice and Bob perform the protocol in Case 3.1 in the new inputs, and the complexity analysis is similar to Case 3.1.
\end{enumerate}
\item \emph{Case 4: $n_1= n-a-b+c $ and $ n_2 = c$, or $n_1= a-c$ and $ n_2 = b-c$, or $n_1= b-c$ and $ n_2 =a-c$.} A similar argument as in Case 3 works. 
\end{itemize}
\end{proof}

%%%%%%%%%%%%%%%%%%%%%%%%%%%%%%%%%%%%%%%%%%%%%%%%%%%%%%%%%%%%%%%%%%%%%%%%%%%%%%%%
%%%%%%%%%%%%%%%%%%%%%%%%%%%%%%%%%%%%%%%%%%%%%%%%%%%%%%%%%%%%%%%%%%%%%%%%%%%%%%%%
%%%%%%%%%%%%%%%%%%%%%%%%%%%%%%%%%%%%%%%%%%%%%%%%%%%%%%%%%%%%%%%%%%%%%%%%%%%%%%%%
\section{Log-Rank Conjecture for Permutation-Invariant Functions}\label{sec:log_rank}
\Cref{th:logrank} states the Log-rank Conjecture for {permutation-invariant} functions. We argue for the lower bound (\Cref{log_rank_lower_bound}) and the upper bound (\Cref{log_rank_upper_bound}) separately. 

\begin{lemma}[Lower Bound]\label{log_rank_lower_bound}
Fix $n\in{\mathbb{Z}^+}$. Let $f:\B^n \times \B^n \rightarrow \inbrace{-1,1}$ be a non-trivial total permutation-invariant function. 
For every $a,b \in [n]$ such that $f_{a,b}$ is not a constant function, we have
\[
\log\rank(f) = \Omega\inparen{\max\inbrace{\log n, \min\inbrace{a,b,n-a,n-b}}},
\]
where $f_{a,b}$ satisfies $f_{a,b}(|x \land y|) = f(x,y)$ for $x,y \in \B^n$.
\end{lemma}

%We prove \Cref{log_rank_lower_bound} below, and postpone the proof of Lemma \ref{log_rank_upper_bound} to Appendix \ref{appendix_log_rank}. 

\begin{proof}%[Proof of \Cref{log_rank_lower_bound}]

We rely on the following two claims to prove the lemma. Two claims show the lower bound on the rank of some special functions respectively.

\begin{fact}[\cite{BW01}, merging Corollary 6 with Lemma 4]\label{BW_rank_lower_bound}
Fix $n\in{\mathbb{Z}^+}$. Let $f:\B^n \times \B^n \rightarrow \inbrace{-1,1}$ be defined as $f(x,y) \coloneqq D(|x\land y|)$ for some predicate $D: \inbrace{0,1,...,n}\rightarrow \inbrace{-1,1}$. If $t$ is the smallest integer such that $D(t) \neq D(t-1)$, then $\log\rank(f) = \Omega\inparen{\log \inparen{\sum_{i=t}^n \binom{n}{i}}}$. 
\end{fact}
\begin{restatable}{fact}{DISJEQ}
\label{disj_eq}
Fix $n\in\mathbb{Z}^+$. 
     Let $\mathcal{X},\mathcal{Y} \coloneqq \inbrace{x\in \B^{n}: |x| = k}$, where $k \le n/2$. Let $\DISJ_n^k:\mathcal{X} \times \mathcal{Y} \rightarrow \inbrace{-1,1}$ and $\Equality_n^k:\mathcal{X} \times \mathcal{Y} \rightarrow \inbrace{-1,1}$ be defined 
    as
\[
\begin{aligned}
\DISJ_n^k(x,y) &\coloneqq
\begin{cases}
-1& \text{ if }|x \land y|  = 0,\\
1& \text{ if } |x \land y| \neq 0,
\end{cases}\\
\Equality_n^k(x,y) &\coloneqq
\begin{cases}
-1& \text{ if }x = y,\\
1& \text{ if } x \neq y.
\end{cases}
\end{aligned}
\]
Then $\rank\inparen{\DISJ_n^k} \ge  \binom{n}{k}-1$ and $\rank\inparen{\Equality_n^k} \ge  \binom{n}{k}-1$.
\end{restatable}
We refer interesting readers to \Cref{rank} for the proof of \Cref{disj_eq}. Now we prove the lemma by casing on the values of $a$ and $b$.

We can assume $a \le b \le n/2$ without loss of generality because the cases where $a > n/2$ or $b > n/2$ can be obtained by flipping each bit of Alice or Bob's input. Thus, it suffices to prove $\log \rank(f) = \Omega\inparen{\max\inbrace{\log n, a}}$.

We prove the following two claims that directly lead to our result:
\begin{enumerate}
\item\label{item:logrank-lower-bound-item1} If $a \le b \le n/2$ and $a = o\inparen{\log n}$, then $\log \rank(f) = \Omega\inparen{\log n}$. 
\item\label{item:logrank-lower-bound-item2} If $a \le b \le n/2$ and $a = \Omega\inparen{\log n}$, then
    $\log \rank(f) = \Omega\inparen{a}$. \label{prop_lowerbound2}
\end{enumerate}

We first prove \Cref{item:logrank-lower-bound-item1}. Suppose $a \le b \le n/2$ and $a = o\inparen{\log n}$. Since $f_{a,b}$ is not a constant function, there exists $c \in [0,a)$ such that $f_{a,b}(c) \neq f_{a,b}(c+1)$. Without loss of generality, we assume $f_{a,b}(c) = -1$.
Let $n' \coloneqq n-(a+b-c-2)$.
Since $b \le n/2$ and $c \le a = o(\log n)$, $n' = n-(a+b-c-2) = \Omega\inparen{n}$. Let $\mathcal{X}$ and $ \mathcal{Y} $ be the set $ \inbrace{x \in \B^{n'}: |x| = 1}$.
For any $x\in\mathcal{X},y\in\mathcal{Y}$,
\[
\DISJ_{n'}^1(x,y) = f_{a,b}(|x\land y|+c) =  f(x',y'),
\]
where
\[
    x' \coloneqq x\underbrace{1\cdots 1}_{c}\underbrace{1\cdots 1}_{a-c-1}\underbrace{0\cdots 0}_{b-c-1} \text{ and }
    y' \coloneqq y\underbrace{1\cdots 1}_{c}\underbrace{0\cdots 0}_{a-c-1}\underbrace{1\cdots 1}_{b-c-1}.
\]
Thus, $\DISJ_{n'}^1$ is a submatrix of $f$.  By \Cref{disj_eq}, we have
\[
\log\rank(f) \ge \log\rank(\DISJ_{n'}^1)  \ge  \log \inparen{n'-1} = \Omega\inparen{\log n}.
\]
   
Now we prove \Cref{item:logrank-lower-bound-item2}. Suppose $a,b \le n/2$ and $\min\inbrace{a,b} = \Omega\inparen{\log n}$,
we consider the following three cases:
    \begin{itemize}
        \item \emph{Case 1: There exists $c \in [4a/7,3a/5)$ such that $f_{a,b}(c) \neq f_{a,b}(c+1)$. } 
        Let $k = \lfloor a/2 \rfloor$ and $k' = \lceil a/2 \rceil$.
        Let $g:\B^k\times \B^k \rightarrow \inbrace{-1,1}$ be such that
    $g(x,y) = f_{a,b}(|x' \land y'|)$
    for every $x,y\in \B^k$, where 
    \[
    \begin{aligned}
    x' &\coloneqq x  \overline{x} \underbrace{0\cdots 0}_k \underbrace{1\cdots 1}_{k'}\underbrace{0\cdots 0}_{b-a}\underbrace{0\cdots 0}_{n-b-2k} \\
    y' &\coloneqq y \underbrace{0\cdots 0}_k \overline{y} \underbrace{1\cdots 1}_{k'}\underbrace{1\cdots 1}_{b-a}\underbrace{0\cdots 0}_{n-b-2k}.
    \end{aligned}
    \]
        Observe that $x',y'\in\B^n$ and $|x'|= a,|y'| =b$.
    Moreover, $g(x,y) = D(|x \land y|)$ for predicate $D: \inbrace{0,1,...,k}\rightarrow \inbrace{-1,1}$ such that $D(z) = f_{a,b}(z+k')$ for every $z \in [0,k]$. Thus, we have $D(c-k') \neq D(c-k'+1)$.
    By \Cref{BW_rank_lower_bound}, we have 
    \[
    \log \rank (g) = \Omega\inparen{\log \inparen{\sum_{i = c-k'+1}^k \binom{k}{i}}}.
    \]
    Since 
    $c-k'+1 < 3a/5-\lceil a/2 \rceil+1 
    \le a/10
    \le k/2$,
    we conclude $\log \rank(g) = \Omega(k) = \Omega(a)$.

    \item \emph{Case 2: There exists $c \in [0,4a/7)$ such that $f_{a,b}(c) \neq f_{a,b}(c+1)$ and $f_{a,b}$ is a constant function in the range $[c,3a/5)$.}  Without loss of generality, we assume $f_{a,b}(c) = -1$.
    Let $l \coloneqq \lfloor 3a/5 \rfloor$, $l' \coloneqq \lceil 2a/5 \rfloor$, $m \coloneqq n-(c+b-a+2l')$. 
Since $a \le b \le n/2$ and $c < 4a/7$, we have 
\[
m = n-(c+b-a+2l') 
\ge 2a - 2l'-c 
= 2l-c 
\ge 2(l-c).
\]
Let $\mathcal{X} $ and $ \mathcal{Y} $ be the set $ \inbrace{x\in \B^{m}: |x| = l-c}$.
For every $x \in \mathcal{X}, y \in \mathcal{Y}$, we have
    \[
    \DISJ_m^{l-c}(x,y) = f_{a,b}(|x' \land y'|) = f(x',y'),
    \]
    where
    \[
    \begin{aligned}
    x' &\coloneqq x\underbrace{1\cdots 1}_c \underbrace{0\cdots 0}_{b-a} \underbrace{0\cdots 0}_{l'}\underbrace{1\cdots 1}_{l'},\\
        y' &\coloneqq y\underbrace{1\cdots 1}_c  \underbrace{1\cdots 1}_{b-a
        }\underbrace{1\cdots 1}_{l'}\underbrace{0\cdots 0}_{l'}.
    \end{aligned}
    \]
    Thus, $\DISJ_m^{l-c}$ is a submatrix of $f$. 
By \Cref{disj_eq},
we have 
\[
\begin{aligned}
\log\rank\inparen{f} &\ge \log\rank(\DISJ_m^{l-c})\\
&=  \Omega\inparen{\log \binom{m}{l-c}} \\
&= \Omega\inparen{l-c} \\
&= \Omega\inparen{a}.
\end{aligned}
\]

    \item \emph{Case 3: There exists $c \in [3a/5,a)$ such that $f_{a,b}(c) \neq f_{a,b}(c+1)$ and $f_{a,b}$ is a constant function in the range $[0,c)$.}  Without loss of generality, we assume $f_{a,b}(c) = -1$. 
Since $a \le b \le \frac{n}{2}$, we have $n-b+c \ge a+c \ge 2c$.
Let $\mathcal{X} $ and $ \mathcal{Y} $ be the set $ \inbrace{x\in \B^{n-b+c}: |x| = c}$.
For every $x \in \mathcal{X}, y \in \mathcal{Y}$, we have
    \[
    \Equality_{n-b+c}^c(x,y) = f_{a,b}(|x' \land y'|) = f(x',y'),
    \]
    where
    \[
    x' \coloneqq x\underbrace{0\cdots 0}_{b-a} \underbrace{0\cdots 0}_{a-c} \text{ and } 
        y' \coloneqq y\underbrace{1\cdots 1}_{b-a}  \underbrace{0\cdots 0}_{a-c
        }.
    \]
    Thus, $\Equality_{n-b+c}^c$ is a submatrix of $f$. By \Cref{disj_eq}, we have 
\[
\begin{aligned}
\log\rank\inparen{f} &\ge \log\rank(\Equality_{n-b+c}^c)\\ 
&= \Omega\inparen{\log \binom{n-b+c}{c}} \\
&= \Omega\inparen{c} \\
&= \Omega\inparen{a}.
\end{aligned}
\]
\end{itemize}
\end{proof}

\begin{restatable}[Upper Bound]{lemma}{LogRankUpperBound}
\label{log_rank_upper_bound}
Fix $n\in{\mathbb{Z}^+}$. Let $f:\B^n \times \B^n \rightarrow \inbrace{-1,1}$ be a non-trivial total permutation-invariant function. 
Then $D(f)$ is
\[
O\inparen{\max_{a,b \in [n]:f_{a,b}\text{ is not constant}}\min\inbrace{a,b,n-a,n-b}\cdot \log n},
\]
where $f_{a,b}$ satisfies $f_{a,b}(|x \land y|) = f(x,y))$ for $x,y\in\B^n$. 
\end{restatable}
\begin{proof}
    We give the following deterministic protocol. 
    First, Alice and Bob exchange the values of $|x|,|y|$ by $O(\log n)$ bits. Suppose $|x| = a, |y| = b$.
    i) If $f_{a,b} = c$ for some constant number $c \in \inbrace{-1,1}$, then they output $c$ directly. ii) If $f_{a,b}$ is not a constant function, then they perform the following operations: if $\binom{n}{a} \le \binom{n}{b}$, Alice sends $x$ to Bob using $\log \binom{n}{a}$ bits, and then Bob outputs $f_{a,b}(x,y)$; otherwise, Bob sends $y$ to Alice using   $\log\binom{n}{b}$ bits, and then Alice outputs $f_{a,b}(x,y)$. 
    In total, the communication cost of the protocol is 
    \begin{equation}\label{eq:cost}
    \log n +\max_{a,b:f_{a,b}\text{ is not constant}}\min\inbrace{\log\binom{n}{a}, \log \binom{n}{b}}
    \end{equation}
    on the worst case. If $f_{a,b}$ is not a constant function, then $0 < a,b < n$, and thus 
    \begin{equation}\label{eq:log_ineq}
    \log n = \log \binom{n}{1} \le \min\inbrace{\log\binom{n}{a}, \log \binom{n}{b}}.
    \end{equation}
    If $a \le n/2$, then $\log \binom{n}{a} \le a \log n$; if $a > n/2$, then $\log \binom{n}{a} = \log \binom{n}{n-a} \le (n-a) \log n$. Thus,
    \begin{equation}\label{eq:log_min}
       \log\binom{n}{a} \le  \min\inbrace{a,n-a}\cdot\log n. 
    \end{equation}
    Combining (\ref{eq:cost}), (\ref{eq:log_ineq}), and (\ref{eq:log_min}), the  communication cost of the protocol is
 \[
 O\inparen{\max_{a,b:f_{a,b}\text{ is not constant}}\min\inbrace{a,b,n-a,n-b}\cdot \log n}.
 \]
\end{proof}

%%%%%%%%%%%%%%%%%%%%%%%%%%%%%%%%%%%%%%%%%%%%%%%%%%%%%%%%%%%%%%%%%%%%%%%%%%%%%%%%
%%%%%%%%%%%%%%%%%%%%%%%%%%%%%%%%%%%%%%%%%%%%%%%%%%%%%%%%%%%%%%%%%%%%%%%%%%%%%%%%
%%%%%%%%%%%%%%%%%%%%%%%%%%%%%%%%%%%%%%%%%%%%%%%%%%%%%%%%%%%%%%%%%%%%%%%%%%%%%%%%
\section{Log-Approximate-Rank Conjecture for Permutation-Invariant Functions}\label{sec:log_arank}

We discuss \Cref{th:PI_log_arank}. In particular, we use the following two lemmas (proved in \Cref{appendix_log_arank}) to prove \Cref{th:PI_log_arank}. Additionally, we note that \Cref{lemma:rank_SecInc,lemma:lower_rank2} are variants of \Cref{le:SecInc,lemma:partial}.

\begin{restatable}{lemma}{rankSecInc}
\label{lemma:rank_SecInc}
Let $n,a,b,c,g\in\mathbb{Z}^+$. The following relations hold:
\begin{enumerate}
    \item $\arank\inparen{\mathrm{ESetInc}_{a, b, c, g}^{n}} \le \arank\inparen{\mathrm{ESetInc}_{a+\ell_1 + \ell_3, b + \ell_2 + \ell_3, c+\ell_3, g}^{n+\ell}}$ for $\ell_1, \ell_2, \ell_3 \ge 0$ such that $\ell_1 + \ell_2 + \ell_3 \le \ell$;
    \item $\arank\inparen{\mathrm{ESetInc}_{a, b, c, g}^{n}} = \arank\inparen{\mathrm{ESetInc}_{a, n-b, a-c, g}^{n}} = \arank\inparen{\mathrm{ESetInc}_{n-a, b, b-c, g}^{n}}$; and 
    \item $\arank\inparen{\mathrm{ESetInc}_{a, b, c, g}^{n}} \le \arank\inparen{\mathrm{ESetInc}_{ka, kb, kc, kg}^{kn}}$ for $k \ge 1$.
\end{enumerate}
\end{restatable}

\begin{restatable}{lemma}{lowerrank}
\label{lemma:lower_rank2}
Fix $k\in\mathbb{Z}$. Let $l$ be a half-integer such that $0 < l\le k/2$. We have
\[
\log\inparen{\arank\inparen{\mathrm{ESetInc}^{4k}_{2k,k,l,1/2}}} = \Omega\inparen{\sqrt{kl}}.
\]
\end{restatable}

\begin{proof}[Proof sketch of \Cref{th:PI_log_arank}]
We use a similar argument as in the proof of \Cref{th:qcc_lowbound}. Namely, for every $a,b \in [n]$ and jump $(c,g) \in \mathscr{J}(f_{a,b})$, let $n_1\coloneqq \min\{[a-c,c,b-c,n-a-b+c]\}$ and $n_2\coloneqq \min\left(\{[a-c,c,b-c,n-a-b+c]\}\setminus \inbrace{n_1}\right)$. We have
\[
\log \arank\inparen{\mathrm{ESetInc}_{a, b, c, g}^{n}} = \Omega\inparen{\frac{\sqrt{n_1n_2}}{g}}.
\]
Since $\mathrm{ESetInc}_{a, b, c, g}^{n}$ is a subfunction of $f$, we have 
\begin{equation*}
\log \arank\inparen{f} =  
\Omega\inparen{\max _{\substack{a, b \in[n] \\(c, g) \in \mathscr{J}\left(f_{a, b}\right)}} \frac{\sqrt{n_{1} n_{2}}}{g}} = \Omega\inparen{m(f)}.
\end{equation*}
Combining \Cref{th:upper} and the above equation, we have \Cref{th:PI_log_arank} as desired.
\end{proof}

\section{Conclusion}\label{sec:con}
This paper proves that the randomized communication complexity of permutation-invariant Boolean functions is at most quadratic of the quantum communication complexity (up to a polylogarithmic factor {of the input size}). Our results suggest that symmetries prevent exponential quantum speedups in communication complexity, extending the analogous research on query complexity. Furthermore, we prove that the Log-rank Conjecture and Log-approximate-rank Conjecture hold for non-trivial permutation-invariant Boolean functions (up to a polylogarithmic factor {of the input size}). There are some interesting problems to explore in the future.
\begin{itemize}
    \item \textit{Permutation invariance over higher alphabets}. In this paper, the permutation-invariant function is a binary function. The interesting question is to generalize our results to larger alphabets, i.e., to permutation-invariant functions of the form $f:\inbrace{0,1,...,m}^n \times \inbrace{0,1,...,m}^n \to \inbrace{-1,1,*}$ where $m \in \mathbb{N}$ and $m > 1$.
    \item \textit{Generalized permutation invariance}. It is possible to generalize our results for a larger class of symmetric functions. One candidate might be a class of functions that have graph-symmetric properties. 
    Suppose $\mathcal{G}_A, \mathcal{G}_B$ are two sets of $n$-vertices graphs, and $G_n$ is a group that acts on the edges of an $n$-vertices graph and permutes them in a way that corresponds to relabeling the vertices of the underlying graph. A function $f:\mathcal{G}_A \times \mathcal{G}_B \to \{-1,1,*\}$ is graph-symmetric if $f(x,y)=f(x\circ \pi, y\circ \pi)$, where $x \in \mathcal{G}_A, y\in\mathcal{G}_B$, and $\pi \in G_n$. 
\end{itemize}

%\section*{Acknowledgments}This work will be presented in the 41st International Symposium on Theoretical Aspects of Computer Science (STACS 2024). Ziyi Guan was partially supported by the Ethereum Foundation. Penghui Yao and Zekun Ye were supported by National Natural Science Foundation of China (Grant No. 62332009, 12347104,  61972191) and Innovation Program for Quantum Science and Technology (Grant No. 2021ZD0302901). 

%\newpage

\appendices

\section*{Appendices Organization}
The appendices are organized as follows. In Section \ref{extend_pre}, we give extended preliminaries. 
Moreover, the relation between the sections of appendices and the omitted proofs are given in \Cref{tab:appendix_organization}. 

\begin{table}[h]
\centering
\caption{The list of omitted proofs.}
\label{tab:appendix_organization}
\begin{tabular}{cl}
\toprule
\textbf{Section} & \textbf{Omitted Proofs} \\
\midrule
\ref{extend_pre} &  Facts \ref{fact:sample_oplus} and \ref{disj_eq} in Section \ref{sec:pre} \\
\ref{appendix_qrcc} &  \Cref{le:SecInc,lemma:partial,lemma:quantum_upper_bound_setinc} in Section \ref{sec:cc} \\
%\ref{appendix_log_rank} & \Cref{log_rank_upper_bound} in Section \ref{sec:log_rank} \\
\ref{appendix_log_arank} &  \Cref{lemma:rank_SecInc,lemma:lower_rank2} in Section \ref{sec:log_arank} \\
\ref{appendix_ghd} & Communication complexity of Gap-Hamming-Distance \\
\bottomrule
\end{tabular}
\end{table}

\section{Extended Preliminaries}\label{extend_pre}

\subsection{Pattern Matrix Method and Approximate Degree}
Pattern matrix method~\cite{She11} is a well-known method for lower bound analysis in quantum communication complexity. Fix $k,n \in \mathbb{Z}$ and $k$ divides $n$. The set $[n]$ is partitioned into $k$ blocks, each consisting of $n/k$ elements. For the universal set $[n]$, let $\mathcal{V}(n,k)$ be a family of subsets that have exactly one element in each block. Clearly, $|\mathcal{V}(n,k)|=(n/k)^k$. Fix $x\in\B^n$ and $V\in\mathcal{V}(n,k)$, let the projection of $x$ onto $V$ be defined as $x|_V \coloneqq (x_{i_1},x_{i_2},\dots,x_{i_k}) \in \B^k$, where $i_1<i_2<\cdots< i_k$ are the elements of $V$. 
\begin{definition}[Pattern matrix~\cite{She11}]\label{def:pm}
   The $(n,k,f)$-pattern matrix $A$ for a function $f \colon \B^k\to \{-1,1,*\}$ is defined as 
\[ A \coloneqq \Big[f(x|_V\oplus
w)\Big]_{x\in\B^n,\,(V,w)\in\mathcal{V}(n,k)\times\B^k}  \;. \]
\end{definition}
Here, $A$ is a matrix of size $2^n$ by $(n/k)^k2^k$, where each row is indexed by strings $x\in\B^n$ and each column is indexed by pairs $(V,w)\in\mathcal{V}(n,k)\times\B^k$. The entries of $A$ are given by $A_{x,(V,w)}= f(x|_V\oplus w)$.

\begin{definition}[Approximate degree]
For $f:\B^n \rightarrow \inbrace{-1,1,*}$ and $0\leq \epsilon < 1$, we say a real multilinear polynomial $p$ approximates $f$ with error $\epsilon$ if:
\begin{enumerate}
\item[(1)]  $|p(x)-f(x)|\leq \epsilon$ for any $x\in \B^n$ such that $f(x) \neq *$;
\item[(2)]  $|p(x)|\leq 1$ for all $x\in\{0,1\}^n$.
\end{enumerate}
The approximate degree of $f$ with error $\epsilon$, denoted by $\widetilde{ \deg}_{\epsilon}(f)$, is the minimum degree among all real multilinear polynomials that approximate $f$ with error $\epsilon$. If $\epsilon = 2/3$, we abbreviate $\widetilde{ \deg}_{\epsilon}(f)$ as $\widetilde{\deg}(f)$. 
\end{definition}

\begin{fact}[\cite{Paturi92}, Theorem 4]\label{fact:paturi}
For symmetric Boolean functions $f:\B^n\rightarrow\inbrace{-1,1,*}$, let $D$ be a Boolean predicate such that $D(|x|) = f(x)$ for any $x \in \B^n$. Then $\adeg\inparen{f} = \Omega\inparen{\sqrt{n(n-\Gamma(D))}}$, where
\[
\begin{aligned}
\Gamma(D) &= \min\{|2k-n+1|:D(k),D(k+1)\neq *, \\
&D(k)\neq D(k+1)\text{ and }0 \le k \le n-1\}.
\end{aligned}
\]
\end{fact}

Each two-party communication problem $F: \mathcal{X} \times \mathcal{Y} \rightarrow \inbrace{-1,1,*}$ can be viewed as a matrix naturally according to the definition of $F$. Then the following fact gives a lower bound on the quantum communication complexity by the approximate degree.
\begin{fact}[\cite{BCG+21}, Theorem 6]
\label{fact:pmm}
Let $F$ be the $(n,t,f)$-pattern matrix, 
where $f\colon \B^t\to\{-1,+1,*\}$ is given.
Then for every $\epsilon\in[0,1)$ and every $\delta<\epsilon/2,$
\begin{align*}
Q_{\delta}(F) &\geq
\frac{1}{4}  \adeg_{\epsilon}(f)\log \left(\frac{n}{t}\right) - 
\frac{1}{2} \log\left(\frac{3}{\epsilon-2\delta}\right).
\end{align*}
\end{fact}

\subsection{Rank}\label{rank}
We first restate and prove Fact \ref{disj_eq} as follows:
\DISJEQ*
\begin{proof}
Let ${\DISJ'}_n^k:\mathcal{X} \times \mathcal{Y} \rightarrow \inbrace{0,1}$ and ${\Equality'}_n^k:\mathcal{X} \times \mathcal{Y} \rightarrow \inbrace{0,1}$ be defined 
    as
\[
{\DISJ'}_n^k(x,y) \coloneqq
\begin{cases}
1,& \text{ if }|x \land y|  = 0,\\
0,& \text{ if } |x \land y| \neq 0.
\end{cases}
\]
and
\[
{\Equality'}_n^k(x,y) \coloneqq
\begin{cases}
1,& \text{ if }x = y,\\
0,& \text{ if } x \neq y.
\end{cases}
\]
    By Example 2.12 in \cite{KN97}, $\rank\inparen{{\DISJ'}_n^k} = \binom{n}{k}$.
    Since $2{\DISJ'}_n^k = J-\DISJ$, where $J$ is the all-ones matrix, we have $\rank\inparen{{\DISJ'}_n^k} \le \rank\inparen{J}+\rank\inparen{\DISJ}$. Thus,
    $\rank\inparen{\DISJ}\ge \rank\inparen{{\DISJ'}_n^k}-\rank\inparen{J} = \binom{n}{k}-1$. 
    Since $\Equality'$ is an identity matrix, we have $\rank\inparen{{\Equality'}_n^k} = \binom{n}{k}$. Similar to $\DISJ$, we have $\rank\inparen{\Equality} \ge \rank\inparen{{\Equality'}_n^k}-1 = \binom{n}{k}-1$.
\end{proof}

\subsection{Approximate Rank}\label{appro_rank}
While some properties of approximate rank have been known for real matrices \cite{KS07, She11}, we generalize them to partial matrices. The proofs are adapted from the original proofs. 

%\subsubsection{Error Reduction}
First, utilizing Facts \ref{poly_const} and \ref{rank_poly}, we prove \Cref{reduction} to show the error reduction of approximate rank.
\begin{fact}[\cite{KS07}, Fact 1]\label{poly_const} 
    Let $0 < E < 1$ be given. Then for each interger $c \ge 1$, there exists a degree-$c$ real univariate polynomial $p(t)$ such that for any $1-E \le |t| \le 1+E$,
    \[
    |p(t)-\sign(t)| \le 8\sqrt{c}\inparen{1-\frac{(1-E)^2}{16}}^c,
    \]
    where
    \[
    \sign(t) \coloneqq
    \begin{cases}
        1,&\text{ if }t > 0,\\
        -1,&\text{ if }t < 0.
    \end{cases}
    \]
\end{fact}
\begin{fact}[\cite{LS09an}, Lemma 11]\label{rank_poly}
    Let $A$ be a real matrix, $p$ be a degree-$d$ polynomial and $B \coloneqq [p(A_{i,j})]_{i,j}$. Then $\rank\inparen{B} \le \inparen{d+1}\rank\inparen{A}^d$.
\end{fact}
%Then we give \Cref{reduction} and prove it as follow:
\begin{restatable}[Error reduction, a generalized version of Corollary 1 in \cite{KS07}]{fact}{reduction}\label{reduction}
    Let $F$ be a matrix with $\inbrace{-1,1,*}$ entries. Let $\epsilon, E$ be constants with $0 < \epsilon < E < 1$. Then $\log\arank_{\epsilon}(F) = O\inparen{\log\arank_E(F)}$.
\end{restatable}
\begin{proof}
Let $c$ be any constant positive integer such that
\[
8\sqrt{c}\inparen{1-\frac{(1-E)^2}{16}}^c \le \epsilon/2.
\]
By \Cref{poly_const}, there is a degree-$d$ polynomial $p(t)$ such that for any $1-E \le |t| \le 1$, 
\[
|p(t)-\sign(t)|\le \epsilon/2.
\]
Let $q(t) = p(t)-\epsilon/2$. Then for any $1-E \le |t| \le 1$, we have
\[
|q(t) - \sign(t)| \le \epsilon, |q(t)| \le 1.
\]
Let $A$ be a real matrix such that $A \in \mathcal{F}_{E}$ and $\rank\inparen{A} = \arank_E\inparen{F}$, where $\mathcal{F}_{E}$ be the set of real matrices that approximates $F$ with error $E$ as \Cref{def:arank}. Then the matrix $B = [q\inparen{A_{i,j}}]_{i,j}$ satisfies that $B \in \mathcal{F}_{\epsilon}.$ By \Cref{rank_poly}, we have $\rank(B) \le (c+1)\rank(A)^c$. Thus,
\[
\begin{aligned}
\arank_{\epsilon}(F) &\le \rank(B) \\
&\le (c+1)\rank(A)^c \\
&= (c+1)\arank_{E}(F)^c.
\end{aligned}
\]
Thus, we have $\log\arank_{\epsilon}(F) = O\inparen{\log\arank_{E}(F)}$.
\end{proof}
%\subsubsection{Lower Bound via Pattern Matrix Method}

Next, %we prove the lower bound via pattern matrix method. 
we first give some useful definitions, facts, and lemmas about the pattern matrix method \cite{She11}.
\begin{definition}[Norm of matrices] 
    For a matrix $A \in \mathbb{R}^{m \times n}$, let the singular values of $A$ be $\sigma_1(A) \ge \sigma_2(A) \ge \cdots \ge \sigma_{\min\inparen{m,n}}(A) \ge 0$.
    The 
    spectral norm and trace norm are given by
    \[
    \begin{aligned}
        ||A|| &\coloneqq \sigma_1(A),\\
        ||A||_{\Sigma} &\coloneqq \sum_i \sigma_i(A).
    \end{aligned}
    \]
\end{definition}
\begin{definition}[Approximate trace norm]\label{app_trace_norm}
    For any matrix $F \in \inbrace{-1,1,*}^{m\times n}$ and $0\leq \epsilon < 1$, the $\epsilon$-approximate trace norm of $F$ is 
    \[
    ||F||_{\Sigma,\epsilon} \coloneqq \min_{A\in\mathcal{F}_{\epsilon}}||A||_{\Sigma},
    \]
    where $\mathcal{F}_{\epsilon}$ is defined in \Cref{def:arank}.
\end{definition}
\begin{definition}\label{def:real_rep}
    For any incomplete matrix $F \in \inbrace{-1,1,*}^{m\times n}$ and $0\leq \epsilon < 1$, we say a real matrix $A$ is a real representation of $F$ if:
\begin{enumerate}
\item[(1)]  $A_{i,j} = F_{i,j}$ for any $i \in [m], j\in [n]$ such that $F_{i,j} \neq *$;
\item[(2)]  $|A_{i,j}|\leq 1$ for all $i\in [m], j\in [n]$.
\end{enumerate}
Let $\mathcal{F}$ be the set of all real representations of $F$.
\end{definition}
By \Cref{def:real_rep,app_trace_norm}, we have 
$||F||_{\Sigma,\epsilon} = \min_{A\in\mathcal{F}}||A||_{\Sigma,\epsilon}$.

\begin{fact}[\cite{She11}, Proposition 2.2]\label{norm_inequality}
    Let $F \in \Re^{m \times n}$ and $\epsilon \ge 0$. Then
    \[
    ||F||_{\Sigma,\epsilon} \ge \sup_{\Psi\in\Re^{m \times n}, ||\Psi|| \neq 0}\frac{\langle F,\Psi\rangle -\epsilon||\Psi||_1}{||\Psi||}.
    \]
\end{fact}
As a corollary of \Cref{norm_inequality}, we have the following lemma:
\begin{lemma}\label{app_trace_norm_lowerbound}
    Let $F \in \inbrace{-1,1,*}^{m\times n}$ and $0 \le \epsilon < 1$. Then
for any $\Psi\in\Re^{m \times n}$ such that $||\Psi||\neq 0$, we have
    \[
    \frac{\sum_{(x,y)\in \dom F} \Psi_{x,y}F_{x,y}-\sum_{(x,y)\notin \dom F}|\Psi_{x,y}|-\epsilon||\Psi||_1}{||\Psi||}
    \]
   is a lower bound of $||F||_{\Sigma,\epsilon}$, where $\dom F := \{(x,y): F_{x,y} \in \{-1,1\}\}$.
\end{lemma}
\begin{proof}
  By \Cref{norm_inequality}, we have
  \[
  \begin{aligned}
  ||F||_{\Sigma,\epsilon} &= \min_{A\in\mathcal{F}}||A||_{\Sigma,\epsilon} \\
  &\ge \min_{A\in\mathcal{F}}\sup_{\Psi\in\Re^{m \times n}, ||\Psi|| \neq 0}\frac{\langle A,\Psi\rangle -\epsilon||\Psi||_1}{||\Psi||}\\
   &= \sup_{\Psi\in\Re^{m \times n}, ||\Psi|| \neq 0}\frac{1}{||\Psi||}\\
   &\inparen{\sum_{(x,y)\in \dom F}\Psi_{x,y}F_{x,y}-\sum_{(x,y)\notin \dom F} |\Psi_{x,y}|-\epsilon||\Psi||_1}.  
%   \frac{\sum_{(x,y)\in \dom F}\Psi_{x,y}F_{x,y}+\sum_{(x,y)\notin \dom F}|\Psi_{x,y}|}{||\Psi||}\\&- \frac{\epsilon||\Psi||_1}{||\Psi||}\\
  \end{aligned}
  \]
\end{proof}
\begin{lemma}\label{arank_fro_lowerbound}
       Let $F \in \inbrace{-1,1,*}^{m\times n}$ and $0 \le \epsilon < 1$. Then
       \[
       \arank_{\epsilon}(F) \ge \frac{||F||_{\Sigma,\epsilon}^2}{\inparen{1+\epsilon}^2mn}.
       \]
\end{lemma}
\begin{proof}
Same as Proposition 2.3 in \cite{She11}, for any real matrix $A$ that approximates $F$ with $\epsilon$, we have
\[
\begin{aligned}
    ||F||_{\Sigma,\epsilon} &=\min_{A \in \mathcal{F}_{\epsilon}}||A||_{\Sigma}\\
    &\le ||A||_{F}\sqrt{\rank{A}}\\
    &\le \sqrt{\sum_{i,j}A_{i,j}^2} \sqrt{\rank{A}} \\
    &\le \sqrt{\inparen{1+\epsilon}^2mn}\sqrt{\rank{A}}.
\end{aligned}
\]
Thus, 
\[
 \arank_{\epsilon}(F)\ge \rank{A} \ge \frac{||F||_{\Sigma,\epsilon}^2}{\inparen{1+\epsilon}^2mn}.
\]
\end{proof}
\begin{fact}[\cite{BCG+21}, Theorem 10]\label{Psi}
        Let $F$ be the $(n,t,f)$-pattern matrix, where $f:\B^t \rightarrow \inbrace{-1,1,*}$ is given. Suppose $d = \adeg_{\epsilon}(f)$. Then there exists $\Psi$ be a $(n,t,2^{-n}(n/t)^{-t}\psi)$-pattern matrix for some function $\psi:\B^t \rightarrow \mathbb{R}$ such that  
        \begin{equation}\label{eq:Psi}
        \begin{aligned}
           ||\Psi||_1 &= 1,\\
            \epsilon &<\sum_{(x,y)\in \dom F}F_{x,y}\Psi_{x,y}-\sum_{(x,y)\notin \dom F}|\Psi_{x,y}|,\\
        ||\Psi|| &\le \inparen{\frac{t}{n}}^{d/2}\inparen{2^{n+t}\inparen{\frac{n}{t}}^t}^{-1/2},
        \end{aligned}
        \end{equation}
        where $\dom F := \{(x,y): F_{x,y} \in \{-1,1\}\}$.
\end{fact}
Finally, we give the following lower bound of approximate rank via the pattern matrix method by \Cref{arank_lowerbound}.
\begin{restatable}[A generalized version of Theorem 1.4 in \cite{She11}]{lemma}{aranklowerbound}\label{arank_lowerbound}
    Let $F$ be the $(n,t,f)$-pattern matrix, where $f:\B^t \rightarrow \inbrace{-1,1,*}$ is given. Then for every $\epsilon \in [0,1)$ and every $\delta \in [0,\epsilon]$,
    \[
    \arank_{\delta}(F) \ge \inparen{\frac{\epsilon-\delta}{1+\delta}}^2\inparen{\frac{n}{t}}^{\adeg_{\epsilon}(f)}.
    \]
\end{restatable}
\begin{proof}
By \Cref{app_trace_norm_lowerbound} and (\ref{eq:Psi}), we have
    \begin{equation}\label{eq:fro_lowerbound}
    ||F||_{\Sigma,\delta} \ge (\epsilon-\delta)\inparen{\frac{n}{t}}^{\adeg_{\epsilon}(f)/2}\sqrt{2^{n+t}(n/t)^t}.    
    \end{equation}
   Since $F$ is a $2^n \times (n/t)^t2^t$ matrix, by \Cref{arank_fro_lowerbound} and (\ref{eq:fro_lowerbound}), we have
    \[
    \begin{aligned}
        \arank_{\delta}(F) 
        &\ge \frac{||F||_{\Sigma,\delta}^2}{\inparen{1+\delta}^2 2^{n+t}(n/t)^t}\\
        &\ge \inparen{\frac{\epsilon-\delta}{1+\delta}}^2\inparen{\frac{n}{t}}^{\adeg_{\epsilon}(f)}.\\
    \end{aligned}
    \]
\end{proof}

\subsection{Sampling in Communication Model}\label{appendix_cc}
First, we give Fact \ref{le:BH95}. Then we restate and prove its randomized version, Fact \ref{fact:sample_oplus}.
\begin{fact}[\cite{brodal1995communication}, Proposition 2]
\label{le:BH95}
Suppose $x,y \in \{0,1\}^n$ are inputs of Alice and Bob such that $|x| \neq |y|$. Alice and Bob can find an index $i$ such that $x_i \neq y_i$ using $O(\log n)$ bits of communication.
\end{fact}
\unisample*
\begin{proof}
The protocol of \Cref{le:BH95} is a deterministic protocol based on binary search. Initially, the search space is $\inbrace{0,1,...,n-1}$. In each round, Alice and Bob shrink the search space to one-half of the original search space. Finally, Alice and Bob find one index $i \in \inbrace{0,1,...,n-1}$ such $x_i \neq y_i$ by $O(\log n)$ rounds. The deterministic protocol can be adjusted to a randomized protocol easily. 
By making the same random permutation to the bits of the inputs using public coins, Alice and Bob can find an index $i$ such that $x_i \neq y_i$ uniformly.
\end{proof}

\subsection{Quantum Amplitude Amplification}
\begin{fact}[\cite{HM19}, Merging Theorem 3 and item (4) in Corollary 4]\label{le:amp_est}
Given a unitary $U$ such that $U\ket{0} = \ket{\psi}$ and an orthogonal projector $\Pi$, there exists a quantum  algorithm outputting an estimation $\tilde{p}$ of $p = \langle \psi | \Pi | \psi \rangle$ such that 
\begin{equation*}
    |\tilde{p}-p| \le \frac{1}{3}\epsilon
\end{equation*}
by $O\inparen{\sqrt{p}/\epsilon}$ calls to (the controlled versions of) $U,U^{\dag}$ and $I-2\Pi$ with bounded error.
\end{fact}
As a direct corollary of \Cref{le:amp_est}, we give the quantum version of \Cref{le:sample} as follows.
\begin{fact}
\label{le:quantum_sample}
Given a set $S$, suppose there is a subset $S'$ of $S$ such that $\frac{|S'|}{|S|} \le \beta-\epsilon$ or $\frac{|S'|}{|S|} \ge \beta+\epsilon$. Suppose we have a quantum sampler such that $\mathcal{S}\ket{0} = \frac{1}{\sqrt{|S|}}\sum_{i \in S}\ket{i}$.
Let orthogonal projector $\Pi$ be defined as $\Pi = \sum_{i \in S'}\ket{i}\bra{i}$. Since $\frac{|S'|}{|S|} = \bra{0}\mathcal{S}^{\dag} \Pi \mathcal{S}\ket{0}$, we can decide whether $\frac{|S'|}{|S|} \le \beta-\epsilon$ or $\frac{|S'|}{|S|} \ge \beta+\epsilon$ by $O\inparen{\sqrt{\beta}/\epsilon}$ calls of (the controlled versions of) $\mathcal{S},\mathcal{S}^{\dag}$ and $I-2\Pi$ with bounded error.
\end{fact}
\begin{remark}\label{re:quantum_sample}
Fact \ref{le:quantum_sample} is also correct if  $\mathcal{S}\ket{0}\ket{0} = \frac{1}{\sqrt{|S|}}\sum_{i \in S}\ket{\psi_i}\ket{i}$ and $\Pi = I \otimes \sum_{i \in S'}\ket{i}\bra{i}$, where $\inbrace{\ket{\psi_i}}$ is an orthonormal basis. 
\end{remark}

\section{Quantum and Randomized Communication Complexity of Permutation-Invariant Functions}\label{appendix_qrcc}
\subsection{The Lower Bound on Quantum Communication Complexity}\label{appendix_qcc_lowerbound}
In this section, we restate and prove the following lemmas:
\QSecInc*
\begin{proof}
\begin{enumerate}
    \item Suppose there exists a quantum protocol $\mathcal{P}$ that solves $\mathrm{ESetInc}_{a+\ell_1 + \ell_3, b + \ell_2 + \ell_3, c+\ell_3, g}^{n+\ell}$. To solve $\mathrm{ESetInc}_{a, b, c, g}^{n}$, Alice and Bob can append some bit strings to their input, respectively. First, Alice appends $\ell_1$ 1's to her input string, and Bob appends $\ell_1$ 0's in the same positions. Second, Bob appends $\ell_2$ 1's to his input, and Alice appends $\ell_2$ 0's in the same positions. Third, Alice and Bob append $\ell_3$ 1's at the same time. Finally, Alice and Bob run the protocol for $\mathrm{ESetInc}_{a+\ell_1 + \ell_3, b + \ell_2 + \ell_3, c+\ell_3, g}^{n+\ell}$.
    \item Suppose there exists a quantum protocol solving $\mathrm{ESetInc}_{a, b, c, g}^{n}$. To solve $\mathrm{ESetInc}_{a, n-b, a-c, g}^{n}$, Bob can first flip each bit of his input, and then Alice and Bob run the protocol for $\mathrm{ESetInc}_{a, b, c, g}^{n}$. Thus $Q\inparen{\mathrm{ESetInc}_{a, n-b, a-c, g}^{n}} \le Q\inparen{\mathrm{ESetInc}_{a, b, c, g}^{n}}$. By a similar argument, we have
   \[
   \begin{aligned}
Q\inparen{\mathrm{ESetInc}_{a, b, c, g}^{n}} &\le Q\inparen{\mathrm{ESetInc}_{a, n-b, a-c, g}^{n}}, \\
Q\inparen{\mathrm{ESetInc}_{a, b, c, g}^{n}} &\le Q\inparen{\mathrm{ESetInc}_{n-a, b, b-c, g}^{n}}, \\
Q\inparen{\mathrm{ESetInc}_{n-a, b, b-c, g}^{n}} &\le 
Q\inparen{\mathrm{ESetInc}_{a, b, c, g}^{n}}.
   \end{aligned}
   \] 
Thus, $Q\inparen{\mathrm{ESetInc}_{a, b, c, g}^{n}} = Q\inparen{\mathrm{ESetInc}_{a, n-b, a-c, g}^{n}} = Q\inparen{\mathrm{ESetInc}_{n-a, b, b-c, g}^{n}}$.
    \item Suppose there exists a quantum protocol solving $\mathrm{ESetInc}_{ka, kb, kc, kg}^{kn}$. To solve $\mathrm{ESetInc}_{a, b, c, g}^{n}$, Alice and Bob can repeat their input strings $k$ times and run the protocol for $\mathrm{ESetInc}_{ka, kb, kc, kg}^{kn}$.
\end{enumerate}
\end{proof}

\lemmapartial*
\begin{proof}
We consider the function $f_{k,l}:\B^k \rightarrow \inbrace{-1,1,*}$ given by 
\begin{align}\label{def:fkl}
f_{k,l}(x)\coloneqq
\begin{cases}
-1, &\text{if $|x|=l-1/2,$}\\
1, &\text{if $|x|=l+1/2,$}\\
*, &\text{otherwise.}
\end{cases}
\end{align}
Let $D$ be a Boolean predicate such that $D(|x|) = f(x)$ for any $x \in \B^k$. By \Cref{fact:paturi}, we have 
\[
\Gamma(D) = |2(l-1/2)-k+1| = |2l-k| = k-2l,
\]
and thus 
\begin{equation}\label{eq:adeg}
\adeg_{1/3}\inparen{f_{k,l}} = \Omega\inparen{\sqrt{k(k-\Gamma(D))}} = \Omega\inparen{\sqrt{kl}}.
\end{equation}
Let $P$ be the $\inparen{2k,k,f_{k,l}}$-pattern matrix defined as \Cref{def:pm}. \Cref{fact:pmm} implies that  $Q_{1/3}\inparen{P} = \Omega\inparen{\sqrt{kl}}$, where 
\begin{equation*}
\begin{aligned}
    P &= \Big[f_{k,l}(x|_V\oplus
w)\Big]_{x\in\B^{2k},\,(V,w)\in\mathcal{V}(2k,k)\times\B^k} \\ 
&= \Big[f_{k,l}(x_1\overline x_1 x_2\overline x_2\ldots
x_{2k}\overline{x_{2k}}|_V)\Big]_{x\in\B^{2k},\,V\in\mathcal{V}(4k,k)}.
\end{aligned}
\end{equation*}
For any $x \in \B^{2k}$, we have $x_1\overline x_1 x_2\overline x_2\ldots
x_{2k}\overline{x_{2k}} \in \B^{4k}$ and $|x_1\overline x_1 x_2\overline x_2\ldots
x_{2k}\overline{x_{2k}}| = 2k$; for any $V\in\mathcal{V}(4k,k)$, we have $V \in \B^{4k}$ and $|V|= k$. Thus, $P$ is a submatrix of $\mathrm{ESetInc}^{4k}_{2k,k,l,1/2}$, which is defined as 
\begin{align*}
&\mathrm{ESetInc}^{4k}_{2k,k,l,1/2}(x, y)\\
&\coloneqq\begin{cases}
-1, &\text{if $|x|=2k,|y|=k$  and  $|x \land y|=l-1/2$}, \\
1, &\text{if $|x|=2k,|y|=k$  and  $|x \land y|=l+1/2$}, \\
*, &\text{otherwise.}
\end{cases}
\end{align*}
As a result, we have 
$Q_{1/3}\inparen{\mathrm{ESetInc}^{4k}_{2k,k,l,1/2}} \ge Q_{1/3}\inparen{P} =  \Omega\inparen{\sqrt{kl}}$.
\end{proof}

\subsection{Quantum Upper Bound of Set-Inclusion Problem}\label{appendix_quantum_speedup}
We restate and prove \Cref{lemma:quantum_upper_bound_setinc} as follows:
\QuanUpperSetInc*
\begin{proof}
The proof of Lemma \ref{lemma:quantum_upper_bound_setinc} 
 is similar to the proof of Lemma \ref{le:extiXandY}. The only essential difference is that  \Cref{le:sample} is replaced by its quantum speedup version, \Cref{le:quantum_sample}.
For simplicity, we discuss the following two cases as the proof of \Cref{le:extiXandY}. For other cases, we can obtain the same result similarly. In the following proof, for $x \in \B^n$, let $S_x \coloneqq \inbrace{i:x_i = 1}$. Let $\Pi \coloneqq \sum_{i \in S_{x \land y}}\ket{i}\bra{i}$.

\begin{enumerate}
    \item \label{item:case1_quantum_upper_bound_setinc}
    $n_1 = c, n_2 = a-c$. In this case, the problem to compute $\mathrm{SetInc}_{a, b, c, g}^{n}(x,y)$ can be reduced equivalently to estimate $\frac{|x \land y|}{|x|}$. We first prove some special unitary operations can be constructed with $O(\log n)$ qubits communication. 
    \begin{enumerate}
        \item Since Alice knows all information about $x$, she can implement $S_x$ such that 
$\mathcal{S}_{x}\ket{0} = \frac{1}{\sqrt{a}}\sum_{i \in S_x}\ket{i}$
by herself without communication. Let $\ket{\phi_1} \coloneqq \mathcal{S}_{x}\ket{0}$. Then we have
\[
\bra{\phi_1}\Pi\ket{\phi_1} =     \frac{|x\land y|}{|x|}.
\]
\item For inputs $x,y \in \B^n$ and any $i \in [n]$, let unitary operator $O_x,O_y,O_{x \land y}$ be defined as
\[
\begin{aligned}
    O_x \ket{i} &\coloneqq (-1)^{x_i}\ket{i},\\
    O_y \ket{i} &\coloneqq (-1)^{y_i}\ket{i},\\
    O_{x \land y}\ket{i} &\coloneqq (-1)^{x_i y_i}\ket{i}.\\
\end{aligned}
\]
Then $I-2\Pi = O_{x \land y}$. For any $\ket{\phi} \coloneqq \sum_{i \in [n]}\alpha_i \ket{i}$, Alice and Bob can perform $I-2\Pi$ to any state $\ket{\phi}$ using $O(\log n)$ qubits of quantum communication as follows: i) Alice performs $O_x$ to $\ket{\phi}$ and sends $\ket{\phi'} = O_x \ket{\phi} = \sum_{i \in [n]}\alpha_i (-1)^{x_i} \ket{i}$ to Bob; ii) Bob performs $O_y$ to $\ket{\phi'}$ and obtains $\sum_{i \in [n]}\alpha_i (-1)^{x_iy_i} \ket{i} = O_{x \land y}\ket{\phi}$.
\end{enumerate}
By \Cref{le:quantum_sample}, a quantum speedup version of 
\Cref{le:sample}, and using an argument similar to Lemma \ref{le:extiXandY}, we can obtain the following conclusion:  using $O\inparen{\frac{\sqrt{n_1n_2}}{g}}$ unitary operations $\mathcal{S}_x$ and $I-2\Pi$, Alice and Bob can estimate $\frac{|x \land y|}{|x|}$ with errors at most $O\inparen{\frac{g}{a}}$. Furthermore, they can decide whether $|x\land y| \ge c+g$ or $|x \land y| \le c-g$ with success probability at least  $1- 1/(6\log n)$ using $O\inparen{\frac{\sqrt{n_1n_2}}{g} \log n \log\log n}$ qubits communication.
\item $n_1 =c$, $n_2 = n-a-b+c$ and $a+b < n$. In this case, the problem to compute $\mathrm{SetInc}_{a, b, c, g}^{n}(x,y)$ can be reduced equivalently to estimate $\frac{|x \land y|}{|\overline{x} \oplus y|}$. Similar to Case \ref{item:case1_quantum_upper_bound_setinc}, we first prove Alice and Bob can construct a quantum sampler $\mathcal{S}_{\overline{x} \oplus y}$ such that
\[
\mathcal{S}_{\overline{x}\oplus y}\ket{0}\ket{0} = \frac{1}{\sqrt{|N|}}\sum_{f \in [N]}\ket{ff}\ket{f(x,y)} = \ket{\phi_2},
\]
using $O\inparen{\log n}$ qubits communication, where $N = C \cdot \frac{n_1n_2}{g^2}$ for some enough large constant $C$ and each $f$ encodes a random function such that $f(x,y)$ is a uniform distribution of the elements in $S_{\overline{x} \oplus y}$.
Specifically, the construction procedure is as follows:
First, Alice and Bob use $O\inparen{\log\frac{n_1n_2}{g^2}}$ qubits communication to transform $\ket{0}$ into a maximum entanglement bipartite state $\frac{1}{\sqrt{N}}\sum_{f \in [N]}\ket{ff}$. Moreover, both Alice and Bob hold one part of the state. Second, by Fact \ref{fact:sample_oplus}, Alice and Bob can sample an element from $S_{\overline{x} \oplus y}$ uniformly using $O(\log n)$ bits of communication. Equivalently, Alice and Bob can generate a random function $f$ using public coins, and output $f(x,y)$ using $O(\log n)$ bits of communication such that the distribution of the output is a uniform distribution of the elements in $S_{\overline{x} \oplus y}$. Since quantum circuits can simulate classical circuits efficiently, Alice and Bob can perform a uniform operation that transforms $\ket{ff}\ket{b}$ into $\ket{ff}\ket{f(x,y)+b}$ for any $f \in N$ using $O\inparen{\log n}$ qubits communication, where both Alice and Bob hold one part of $\ket{ff}$,
$b \in \inbrace{0,1,...,n-1}$ and the addition is with modulo $n$. As a whole, Alice and Bob can construct the above quantum sampler $\mathcal{S}_{\overline{x}\oplus y}$ using $O\inparen{\log \frac{n_1n_2}{g^2} + \log n} = O\inparen{\log n}$ qubits communication.
Furthermore, since $f(x,y)$ is a uniform distribution of the elements in $S_{\overline{x} \oplus y}$ for any $f \in \inbrace{0,1,...,N-1}$,
$N = C \cdot \frac{n_1n_2}{g^2}$ for some enough large constant $C$, and 
\[
\bra{\phi_2}\inparen{I \otimes \Pi}\ket{\phi_2} = \frac{|\inbrace{f:f(x,y) \in S_{x \land y}}|}{N},
\]
we have 
\begin{equation}\label{eq:innerproduct_phi_two}
\left|\bra{\phi_2}\inparen{I \otimes \Pi}\ket{\phi_2}-\frac{|x\land y|}{|\overline{x} \oplus y|}\right|
  = O\inparen{\frac{g}{m}},
\end{equation}
similar to the analysis of Case \ref{a_b_less_n}  of Lemma \ref{le:extiXandY},
where $m = n_1+n_2$.
\end{enumerate}
By \Cref{le:quantum_sample}, a quantum speedup version of 
\Cref{le:sample}, and using an argument similar to Lemma \ref{le:extiXandY}, we can obtain the following conclusion,
using $O\inparen{\frac{\sqrt{n_1n_2}}{g}}$ unitary operations $\mathcal{S}_{|\overline{x} \oplus y|}$ and $I-2\Pi$, Alice and Bob can estimate  $\bra{\phi_2}\inparen{I \otimes \Pi}\ket{\phi_2}$ with errors at most $O\inparen{\frac{g}{m}}$. By (\ref{eq:innerproduct_phi_two}), Alice and Bob also can estimate $\frac{|x \land y|}{|\overline{x} \oplus y|}$ with errors at most $O\inparen{\frac{g}{m}}$. Next, using the same argument as \Cref{le:extiXandY}, we can prove that Alice and Bob can decide whether $|x\land y| \ge c+g$ or $|x \land y| \le c-g$ with success probability at least $1- 1/(6\log n)$ using $O\inparen{\frac{\sqrt{n_1n_2}}{g} \log n \log\log n}$ qubits communication.
\end{proof}

%\section{Log-Rank Conjecture of Permutation-Invariant Functions}\label{appendix_log_rank}
%In this section, we restate and prove Lemma \ref{log_rank_upper_bound}.
%\LogRankUpperBound*

\section{Log-Approximate-Rank Conjecture of Permutation-Invariant Functions}\label{appendix_log_arank}
In this section, we restate and prove \Cref{lemma:rank_SecInc,lemma:lower_rank2}.
\rankSecInc*
\begin{proof}
By the definition of approximate rank (\Cref{def:arank}), to prove $\arank(A) \le \arank(B)$, it suffices to prove $A$ is a submatrix of $B$. Then we discuss three cases:
    \begin{enumerate}
        \item 
        For $x,y \in \B^n$,
        let 
        \[
        \begin{aligned}
            x' &\coloneqq x\underbrace{1\cdots 1}_{l_1}\underbrace{0\cdots 0}_{l_2}\underbrace{1\cdots 1}_{l_3}\underbrace{0\cdots 0}_{l-\inparen{l_1+l_2+l_3}},\\
            y' &\coloneqq y\underbrace{0\cdots 0}_{l_1}\underbrace{1\cdots 1}_{l_2}\underbrace{1\cdots 1}_{l_3}\underbrace{0\cdots 0}_{l-\inparen{l_1+l_2+l_3}}.\\
        \end{aligned}
        \]
        Then 
        $\mathrm{ESetInc}_{a, b, c, g}^{n}\inparen{x,y} = \mathrm{ESetInc}_{a+\ell_1 + \ell_3, b + \ell_2 + \ell_3, c+\ell_3, g}^{n+\ell}\inparen{x',y'}$.
        Thus, $\mathrm{ESetInc}_{a, b, c, g}^{n}$ is a submatrix of $\mathrm{ESetInc}_{a+\ell_1 + \ell_3, b + \ell_2 + \ell_3, c+\ell_3, g}^{n+\ell}$.
        \item For $x,y\in \B^n$, we have 
        \[
        \begin{aligned}
    \mathrm{ESetInc}_{a, b, c, g}^{n}(x,y) &= \mathrm{ESetInc}_{a, n-b, a-c, g}^{n}(x,\overline{y})\\
    &= \mathrm{ESetInc}_{n-a, b, b-c, g}^{n}(\overline{x},y).
    \end{aligned}
        \]
        \item Since $\mathrm{ESetInc}_{a, b, c, g}^{n}(x,y) =  \mathrm{ESetInc}_{ka, kb, kc, kg}^{kn}(\underbrace{x\cdots
        x}_{k},\underbrace{y\cdots y}_k)$, we have $\mathrm{ESetInc}_{a, b, c, g}^{n}$ is a submatrix of $\mathrm{ESetInc}_{ka, kb, kc, kg}^{kn}$.
    \end{enumerate}
\end{proof}
\lowerrank*
\begin{proof}
Let $P$ be the $\inparen{2k,k,f_{k,l}}$-pattern matrix, where $f_{k,l}$ is defined in (\ref{def:fkl}). By \Cref{lemma:partial}, $P$ is a submatrix of  $\mathrm{ESetInc}^{4k}_{2k,k,l,1/2}$. Thus, we have
\[
\log\inparen{\arank\inparen{\mathrm{ESetInc}^{4k}_{2k,k,l,1/2}}} \ge \log\inparen{\arank\inparen{P}}.
\]
By Facts \ref{reduction} and \ref{arank_lowerbound} and (\ref{eq:adeg}), we have
\[
\log\inparen{\arank\inparen{P}} =  
\Omega\inparen{\adeg\inparen{f_{k,l}}}.
\]
Thus, 
\[
\begin{aligned}
\log\inparen{\arank\inparen{\mathrm{ESetInc}^{4k}_{2k,k,l,1/2}}} &= \Omega\inparen{\adeg\inparen{f_{k,l}}} \\
&= \Omega\inparen{\sqrt{kl}},
\end{aligned}
\]
where the second equality comes from (\ref{eq:adeg}).
\end{proof}
\section{Communciation Complexity of Gap-Hamming-Distance Problem}\label{appendix_ghd}
Ref. \cite{GKS16} gave the lower and upper bounds on the randomized communication complexity of Gap-Hamming-Distance as \Cref{lemma_gks_lower,lemma_gks_upper}. We show \Cref{lemma_gks_lower,lemma_gks_upper} can be expressed as \Cref{lemma_gks_lower_2,lemma_gks_upper_2} equivalently.
\begin{lemma}[Lemma 3.3 in \cite{GKS16}]\label{lemma_gks_lower}
Fix $n\in{\mathbb{Z^+}}$. Consider $a,b \in \inbrace{1,...,n-1}$ and $c-g,c+g$ are achievable Hamming distances of $\Delta(x,y)$ when $|x| = a, |y| = b$.  Then 
	\begin{equation}\label{gks_upper}
    \begin{aligned}
		R(\GHD_{a,b,c,g}^n) &= \Omega\inparen{\frac{\min\inbrace{a,b,c,n-a,n-b,n-c}}{g}},\\
        R(\GHD_{a,b,c,g}^n) &= \Omega\inparen{\log\inparen{\frac{\min\inbrace{c,n-c}}{g}}}.
%        \Omega\inparen{\max\inbrace{\inparen{\frac{\min\inbrace{a,b,c,n-a,n-b,n-c}}{g}},\log\inparen{\frac{\min\inbrace{c,n-c}}{g}}}}. 
	\end{aligned}
    \end{equation}
\end{lemma}

\begin{lemma}\label{lemma_gks_lower_2}
Fix $n\in{\mathbb{Z^+}}$. Consider $a,b \in \inbrace{1,...,n-1}$ and $c-g,c+g$ are achievable Hamming weights of $|x \land y|$ when $|x| = a, |y| = b$. 
Let $n_1\coloneqq \min\{[a-c,c,b-c,n-a-b+c]\}$ and $n_2\coloneqq \min\left(\{[a-c,c,b-c,n-a-b+c]\}\setminus \inbrace{n_1}\right)$. Then $R(\SetInc_{a,b,c,g}^n)$ has two lower bounds: $\Omega\inparen{\frac{n_2}{g}}$ and
\[
%\begin{aligned}
%R(\SetInc_{a,b,c,g}^n) &=\Omega\inparen{\frac{n_2}{g}}, \\
%R(\SetInc_{a,b,c,g}^n) &=
\Omega\inparen{\log\frac{\min\inbrace{a+b-2c,n-a-b+2c}}{g}}.
%\end{aligned}
\]
\end{lemma}
\begin{proof}
	By \Cref{def:GHD,def:setinc}, $\SetInc_{a,b,c,d}^n$ and $\GHD_{a,b,a+b-2c,2g}^n$ are the same problems. Thus, we have 
    \[
    R(\SetInc_{a,b,c,d}^n) = R(\GHD_{a,b,a+b-2c,2g}^n). 
	\]
    By \Cref{lemma_gks_lower}, $R(\GHD_{a,b,a+b-2c,2g}^n)$ has two lower bounds:
    \[
	%\begin{aligned}
   \Omega\inparen{\frac{\min\inbrace{a,b,a+b-2c,n-a,n-b,n-(a+b-2c)}}{g}},\\
	%\end{aligned}
	\]
    and
        \[
    \Omega\inparen{\log\frac{\min\inbrace{a+b-2c,n-a-b+2c}}{g}}.
    \]
	Since 
	\begin{equation*}
		\begin{aligned}
			a &= (a-c)+c, \\ 
			%|x| = |x \land y|+|x \land \overline{y}| \\
			b &= (b-c)+c,\\
			%|y| = |x \land y|+|\overline{x} \land y|\\
			a+b-2c &= (a-c)+(b-c),\\
			n-a &= (n-a-b+c)+(b-c), \\
			n-b &= (n-a-b+c)+(a-c),\\
			n-(a+b-2c) &= (n-a-b+c)+c, \\   
		\end{aligned}
	\end{equation*}
	and $n_1,n_2$ are smallest two numbers in $a-c,c,b-c,n-a-b+c$,
	we have
	\[
    \begin{aligned}
&\min\inbrace{a,b,a+b-2c,n-a,n-b,n-(a+b-2c)} \\
&= n_1+n_2\\
&\ge n_2.
	\end{aligned}
    \]
%	Thus, we have
 %   \[
%\begin{aligned}
%R(\SetInc_{a,b,c,g}^n) &=\Omega\inparen{\frac{n_2}{g}}, \\
%R(\SetInc_{a,b,c,g}^n) &=\Omega\inparen{\log\frac{\min\inbrace{a+b-2c,n-a-b+2c}}{g}}.
%\end{aligned}
%\]
\end{proof} 
 
\begin{lemma}[Lemma 3.4 in \cite{GKS16}]\label{lemma_gks_upper}
Fix $n\in{\mathbb{Z^+}}$. Consider $a,b \in \inbrace{1,...,n-1}$ and $c-g,c+g$ are achievable Hamming distances of $\Delta(x,y)$ when $|x| = a, |y| = b$. If $a \le b \le n/2$, then $R(\GHD_{a,b,c,g}^n)$ is
	\begin{equation}
		\label{gks_lower}
		 O\inparen{\min\inbrace{\inparen{\frac{a}{g}}^2\log \inparen{\frac{b}{g}},\inparen{\frac{c}{g}}^2,\inparen{\frac{n-c}{g}}^2}}.
	\end{equation}
\end{lemma}

\begin{lemma}\label{lemma_gks_upper_2}
Fix $n\in{\mathbb{Z^+}}$. Consider $a,b \in \inbrace{1,...,n-1}$ and $c-g,c+g$ are achievable Hamming weights of $|x \land y|$ when $|x| = a, |y| = b$. If $a \le b \le n/2$, then
\[		R(\SetInc_{a,b,c,g}^n) = O\inparen{\inparen{\frac{n_2}{g}}^2\log \frac{b}{g}}.
\]
\end{lemma}
\begin{proof}
Since $a \le b \le n/2$,	 similar to the proof of \Cref{lemma_gks_lower_2}, we have
	\[
	\begin{aligned}
	&\min\inbrace{a,a+b-2c,n-(a+b-2c)}\\
    &= 
	\min\inbrace{a,b,a+b-2c,n-a,n-b,n-(a+b-2c)} \\
	&=n_1+n_2 \\
	&\le 2n_2.
	\end{aligned}
	\]
	 By \Cref{lemma_gks_upper}, we have
	 	\[
	\begin{aligned}
	&R(\SetInc_{a,b,c,d}^n) \\
    &= R(\GHD_{a,b,a+b-2c,2g}^n) \\
	&=  O\inparen{\frac{\min\inbrace{a,a+b-2c,n-(a+b-2c)}^2}{(2g)^2}\log \inparen{\frac{b}{2g}}}\\
 %   \min\inbrace{\inparen{\frac{a}{2g}}^2,\inparen{\frac{a+b-2c}{2g}}^2,\inparen{\frac{n-(a+b-2c)}{2g}}^2}\cdot \log \inparen{\frac{b}{2g}}}\\
	&= O\inparen{\inparen{\frac{n_2}{g}}^2\log \inparen{\frac{b}{g}}}.\\
		\end{aligned}
	\]
\end{proof}

\section*{Acknowledgement}
	The authors  would like to thank the editor and anonymous reviewers for
their helpful comments and suggestions.
We also would like to express our gratitute to Arkadey Chattopadhyay for drawing our attention to the proper attribution of the resolution of the Log-Approximate-Rank Conjecture for $\mathrm{XOR}$-symmetric functions. Accordingly, we have revised the manuscript to correctly acknowledge the work of Chattopadhyay and Mande~\cite{CM17}, who first established this result.
%We are grateful to Arkadev Chattopadhyay for bringing to our attention the need to correct the attribution for the resolution of the Log-Approximate-Rank Conjecture for symmetric XOR functions. In our original manuscript, we incorrectly credited this result to Zhang and Shi~\cite{ZS09}. We have corrected the reference to the work of Chattopadhyay and Mande~\cite{CM17}, which first established this result.
%We are grateful to Arkadev Chattopadhyay for bringing to our attention the need to correct the attribution for the resolution of the Log-Approximate-Rank Conjecture for symmetric XOR functions. 

%Upon further investigation, we acknowledge that the first complete proof was in fact provided by Chattopadhyay, Mande, and Sherif in their FSTTCS 2017 paper, "A lifting theorem with applications to symmetric functions" [citation]. We sincerely thank them for their pivotal contribution and have updated our manuscript accordingly. We also thank the colleague who pointed us to the correct reference.

%We sincerely thank Arkadev Chattopadhyay for pointing out an error in our attribution regarding the proof of the Log-Approximate-Rank Conjecture for symmetric XOR functions. We have corrected the reference from Zhang and Shi~\cite{ZS09} to the work of Chattopadhyay and Mande~\cite{CM17}, which first established this result.

\bibliography{ref}

\begin{thebibliography}{10}

\bibitem{AA14}
Scott Aaronson and Andris Ambainis.
\newblock The need for structure in quantum speedups.
\newblock {\em Theory of Computing}, 10:133--166, 2014.
\newblock \href {https://doi.org/10.4086/toc.2014.v010a006} {\path{doi:10.4086/toc.2014.v010a006}}.

\bibitem{AB16}
Scott Aaronson and Shalev Ben{-}David.
\newblock Sculpting quantum speedups.
\newblock In {\em Proceedings of the 31st Conference on Computational Complexity}, volume~50, pages 26:1--26:28, 2016.
\newblock \href {https://doi.org/10.4230/LIPIcs.CCC.2016.26} {\path{doi:10.4230/LIPIcs.CCC.2016.26}}.

\bibitem{AS04}
Scott Aaronson and Yaoyun Shi.
\newblock Quantum lower bounds for the collision and the element distinctness problems.
\newblock {\em Journal of {ACM}}, 51(4):595--605, 2004.
\newblock \href {https://doi.org/10.1145/1008731.1008735} {\path{doi:10.1145/1008731.1008735}}.

\bibitem{8948623}
Anurag Anshu, Naresh~Goud Boddu, and Dave Touchette.
\newblock Quantum log-approximate-rank conjecture is also false.
\newblock In {\em 2019 IEEE 60th Annual Symposium on Foundations of Computer Science (FOCS)}, pages 982--994, 2019.
\newblock \href {https://doi.org/10.1109/FOCS.2019.00063} {\path{doi:10.1109/FOCS.2019.00063}}.

\bibitem{ATYY17}
Anurag Anshu, Dave Touchette, Penghui Yao, and Nengkun Yu.
\newblock Exponential separation of quantum communication and classical information.
\newblock In {\em Proceedings of the 49th Annual {ACM} {SIGACT} Symposium on Theory of Computing}, pages 277--288. {ACM}, 2017.
\newblock \href {https://doi.org/10.1145/3055399.3055401} {\path{doi:10.1145/3055399.3055401}}.

\bibitem{bar2008exponential}
Ziv Bar{-}Yossef, T.~S. Jayram, and Iordanis Kerenidis.
\newblock Exponential separation of quantum and classical one-way communication complexity.
\newblock {\em {SIAM} J. Comput.}, 38(1):366--384, 2008.
\newblock \href {https://doi.org/10.1137/060651835} {\path{doi:10.1137/060651835}}.

\bibitem{BBC+01}
Robert Beals, Harry Buhrman, Richard Cleve, Michele Mosca, and Ronald de~Wolf.
\newblock Quantum lower bounds by polynomials.
\newblock {\em Journal of the ACM}, 48(4):778--797, 2001.
\newblock \href {https://doi.org/10.1145/502090.502097} {\path{doi:10.1145/502090.502097}}.

\bibitem{BCG+21}
Aleksandrs Belovs, Arturo Castellanos, Fran{\c{c}}ois~Le Gall, Guillaume Malod, and Alexander~A. Sherstov.
\newblock Quantum communication complexity of distribution testing.
\newblock {\em Quantum Information and Computation}, 21(15{\&}16):1261--1273, 2021.
\newblock \href {https://doi.org/10.26421/QIC21.15-16-1} {\path{doi:10.26421/QIC21.15-16-1}}.

\bibitem{BenDavid16}
Shalev Ben{-}David.
\newblock The structure of promises in quantum speedups.
\newblock In {\em Proceedings of the 11th Conference on the Theory of Quantum Computation, Communication and Cryptography}, volume~61, pages 7:1--7:14, 2016.
\newblock \href {https://doi.org/10.4230/LIPIcs.TQC.2016.7} {\path{doi:10.4230/LIPIcs.TQC.2016.7}}.

\bibitem{BCG+20}
Shalev Ben{-}David, Andrew~M. Childs, Andr{\'{a}}s Gily{\'{e}}n, William Kretschmer, Supartha Podder, and Daochen Wang.
\newblock Symmetries, graph properties, and quantum speedups.
\newblock In {\em Proceedings of the 61st {IEEE} Annual Symposium on Foundations of Computer Science}, pages 649--660, 2020.
\newblock \href {https://doi.org/10.1109/FOCS46700.2020.00066} {\path{doi:10.1109/FOCS46700.2020.00066}}.

\bibitem{10.1145/2629598}
Eli Ben-Sasson, Shachar Lovett, and Noga Ron-Zewi.
\newblock An additive combinatorics approach relating rank to communication complexity.
\newblock {\em Journal of the ACM}, 61(4), 2014.
\newblock \href {https://doi.org/10.1145/2629598} {\path{doi:10.1145/2629598}}.

\bibitem{brassard2002quantum}
Gilles Brassard, Peter H{\o}yer, Michele Mosca, and Alain Tapp.
\newblock Quantum amplitude amplification and estimation.
\newblock {\em Contemporary Mathematics}, 305:53--74, 2002.

\bibitem{BGKT19}
Sergey Bravyi, David Gosset, Robert K{\"{o}}nig, and Marco Tomamichel.
\newblock Quantum advantage with noisy shallow circuits in {3D}.
\newblock In {\em Proccedings of the 60th {IEEE} Annual Symposium on Foundations of Computer Science}, pages 995--999, 2019.
\newblock \href {https://doi.org/10.1109/FOCS.2019.00064} {\path{doi:10.1109/FOCS.2019.00064}}.

\bibitem{brodal1995communication}
Gerth~St{\o}lting Brodal and Thore Husfeldt.
\newblock {\em A communication complexity proof that symmetric functions have logarithmic depth}.
\newblock BRICS, 1995.

\bibitem{BW01}
Harry Buhrman and Ronald de~Wolf.
\newblock Communication complexity lower bounds by polynomials.
\newblock In {\em Proceedings of the 16th Annual {IEEE} Conference on Computational Complexity}, pages 120--130, 2001.
\newblock \href {https://doi.org/10.1109/CCC.2001.933879} {\path{doi:10.1109/CCC.2001.933879}}.

\bibitem{Cha19}
Andr{\'{e}} Chailloux.
\newblock A note on the quantum query complexity of permutation symmetric functions.
\newblock In {\em Proceedings of the 10th Innovations in Theoretical Computer Science Conference}, volume 124, pages 19:1--19:7, 2019.
\newblock \href {https://doi.org/10.4230/LIPIcs.ITCS.2019.19} {\path{doi:10.4230/LIPIcs.ITCS.2019.19}}.

\bibitem{CR12}
Amit Chakrabarti and Oded Regev.
\newblock An optimal lower bound on the communication complexity of {G}ap-{H}amming-{D}istance.
\newblock {\em SIAM Journal on Computing}, 41(5):1299--1317, 2012.
\newblock \href {https://doi.org/10.1137/120861072} {\path{doi:10.1137/120861072}}.

\bibitem{CM17}
Arkadev Chattopadhyay and Nikhil~S. Mande.
\newblock A lifting theorem with applications to symmetric functions.
\newblock In {\em Proceedings of the 37th {IARCS} Annual Conference on Foundations of Software Technology and Theoretical Computer Science, {FSTTCS} 2017}, volume~93, pages 23:1--23:14, 2017.
\newblock \href {https://doi.org/10.4230/LIPICS.FSTTCS.2017.23} {\path{doi:10.4230/LIPICS.FSTTCS.2017.23}}.

\bibitem{10.1145/3396695}
Arkadev Chattopadhyay, Nikhil~S. Mande, and Suhail Sherif.
\newblock The log-approximate-rank conjecture is false.
\newblock {\em J. ACM}, 67(4), jun 2020.
\newblock \href {https://doi.org/10.1145/3396695} {\path{doi:10.1145/3396695}}.

\bibitem{CCHL21}
Sitan Chen, Jordan Cotler, Hsin{-}Yuan Huang, and Jerry Li.
\newblock Exponential separations between learning with and without quantum memory.
\newblock In {\em Proccedings of the 62nd {IEEE} Annual Symposium on Foundations of Computer Science}, pages 574--585, 2021.
\newblock \href {https://doi.org/10.1109/FOCS52979.2021.00063} {\path{doi:10.1109/FOCS52979.2021.00063}}.

\bibitem{GKKRW07}
Dmitry Gavinsky, Julia Kempe, Iordanis Kerenidis, Ran Raz, and Ronald de~Wolf.
\newblock Exponential separations for one-way quantum communication complexity, with applications to cryptography.
\newblock In {\em Proceedings of the 39th Annual {ACM} Symposium on Theory of Computing}, pages 516--525. {ACM}, 2007.
\newblock \href {https://doi.org/10.1145/1250790.1250866} {\path{doi:10.1145/1250790.1250866}}.

\bibitem{GP08}
Dmitry Gavinsky and Pavel Pudl{\'{a}}k.
\newblock Exponential separation of quantum and classical non-interactive multi-party communication complexity.
\newblock In {\em Proceedings of the 23rd Annual {IEEE} Conference on Computational Complexity}, pages 332--339. {IEEE} Computer Society, 2008.
\newblock \href {https://doi.org/10.1109/CCC.2008.27} {\path{doi:10.1109/CCC.2008.27}}.

\bibitem{GKS16}
Badih Ghazi, Pritish Kamath, and Madhu Sudan.
\newblock Communication complexity of permutation-invariant functions.
\newblock In {\em Proceedings of the 27th Annual {ACM-SIAM} Symposium on Discrete Algorithms}, pages 1902--1921, 2016.
\newblock \href {https://doi.org/10.1137/1.9781611974331.ch134} {\path{doi:10.1137/1.9781611974331.ch134}}.

\bibitem{GS20}
Daniel Grier and Luke Schaeffer.
\newblock Interactive shallow clifford circuits: quantum advantage against {NC}{\({^1}\)} and beyond.
\newblock In {\em Proccedings of the 52nd Annual {ACM} {SIGACT} Symposium on Theory of Computing}, pages 875--888, 2020.
\newblock \href {https://doi.org/10.1145/3357713.3384332} {\path{doi:10.1145/3357713.3384332}}.

\bibitem{Grover96}
L.~K. Grover.
\newblock A fast quantum mechanical algorithm for database search.
\newblock In {\em Proceedings of the 28th {IEEE} Annual Symposium on Theory of Computing}, pages 212--219, 1996.
\newblock \href {https://doi.org/10.1145/237814.237866} {\path{doi:10.1145/237814.237866}}.

\bibitem{HM19}
Yassine Hamoudi and Fr{\'{e}}d{\'{e}}ric Magniez.
\newblock Quantum chebyshev's inequality and applications.
\newblock In {\em Proceedings of the 46th International Colloquium on Automata, Languages, and Programming}, volume 132, pages 69:1--69:16, 2019.
\newblock \href {https://doi.org/10.4230/LIPIcs.ICALP.2019.69} {\path{doi:10.4230/LIPIcs.ICALP.2019.69}}.

\bibitem{HSZZ06}
Wei Huang, Yaoyun Shi, Shengyu Zhang, and Yufan Zhu.
\newblock The communication complexity of the {H}amming distance problem.
\newblock {\em Information Processing Letter}, 99(4):149--153, 2006.
\newblock \href {https://doi.org/10.1016/j.ipl.2006.01.014} {\path{doi:10.1016/j.ipl.2006.01.014}}.

\bibitem{Kallaugher21}
John Kallaugher.
\newblock A quantum advantage for a natural streaming problem.
\newblock In {\em Proccedings of the 62nd {IEEE} Annual Symposium on Foundations of Computer Science}, pages 897--908, 2021.
\newblock \href {https://doi.org/10.1109/FOCS52979.2021.00091} {\path{doi:10.1109/FOCS52979.2021.00091}}.

\bibitem{KS07}
Adam~R. Klivans and Alexander~A. Sherstov.
\newblock A lower bound for agnostically learning disjunctions.
\newblock In {\em Proceeding of the 20th Annual Conference on Learning Theory}, volume 4539, pages 409--423. Springer, 2007.
\newblock \href {https://doi.org/10.1007/978-3-540-72927-3\_30} {\path{doi:10.1007/978-3-540-72927-3\_30}}.

\bibitem{10.1145/3406325.345099}
Alexander Knop, Shachar Lovett, Sam McGuire, and Weiqiang Yuan.
\newblock Log-rank and lifting for {AND}-functions.
\newblock In {\em Proceedings of the 53rd Annual ACM SIGACT Symposium on Theory of Computing}, page 197–208, 2021.
\newblock \href {https://doi.org/10.1145/3406325.3450999} {\path{doi:10.1145/3406325.3450999}}.

\bibitem{KN97}
Eyal Kushilevitz and Noam Nisan.
\newblock {\em Communication complexity}.
\newblock Cambridge University Press, 1997.

\bibitem{LS09an}
Troy Lee and Adi Shraibman.
\newblock An approximation algorithm for approximation rank.
\newblock In {\em Proceedings of the 24th Annual {IEEE} Conference on Computational Complexity}, pages 351--357. {IEEE} Computer Society, 2009.
\newblock \href {https://doi.org/10.1109/CCC.2009.25} {\path{doi:10.1109/CCC.2009.25}}.

\bibitem{lee2009lower}
Troy Lee and Adi Shraibman.
\newblock Lower bounds in communication complexity.
\newblock {\em Foundations and Trends in Theoretical Computer Science}, 3(4):263--398, 2009.
\newblock \href {https://doi.org/10.1561/0400000040} {\path{doi:10.1561/0400000040}}.

\bibitem{LS88}
L.~Lovasz and M.~Saks.
\newblock Lattices, mobius functions and communications complexity.
\newblock In {\em Proceedings of the 29th Annual Symposium on Foundations of Computer Science}, pages 81--90, 1988.
\newblock \href {https://doi.org/10.1109/SFCS.1988.21924} {\path{doi:10.1109/SFCS.1988.21924}}.

\bibitem{10.1145/2724704}
Shachar Lovett.
\newblock Communication is bounded by root of rank.
\newblock {\em Journal of the ACM}, 63(1), 2016.
\newblock \href {https://doi.org/10.1145/2724704} {\path{doi:10.1145/2724704}}.

\bibitem{Montanaro11}
Ashley Montanaro.
\newblock A new exponential separation between quantum and classical one-way communication complexity.
\newblock {\em Quantum Information and Computation}, 11(7{\&}8):574--591, 2011.
\newblock \href {https://doi.org/10.26421/QIC11.7-8-3} {\path{doi:10.26421/QIC11.7-8-3}}.

\bibitem{Paturi92}
Ramamohan Paturi.
\newblock On the degree of polynomials that approximate symmetric boolean functions (preliminary version).
\newblock In {\em Proceedings of the 24th Annual {ACM} Symposium on Theory of Computing}, pages 468--474. {ACM}, 1992.
\newblock \href {https://doi.org/10.1145/129712.129758} {\path{doi:10.1145/129712.129758}}.

\bibitem{raz1999exponential}
Ran Raz.
\newblock Exponential separation of quantum and classical communication complexity.
\newblock In {\em Proceedings of the Thirty-First Annual {ACM} Symposium on Theory of Computing}, pages 358--367. {ACM}, 1999.
\newblock \href {https://doi.org/10.1145/301250.301343} {\path{doi:10.1145/301250.301343}}.

\bibitem{razborov2003quantum}
Alexander~A Razborov.
\newblock Quantum communication complexity of symmetric predicates.
\newblock {\em Izvestiya: Mathematics}, 67(1):145, 2003.
\newblock \href {https://doi.org/10.1070/IM2003v067n01ABEH000422} {\path{doi:10.1070/IM2003v067n01ABEH000422}}.

\bibitem{She11}
Alexander~A. Sherstov.
\newblock The pattern matrix method.
\newblock {\em {SIAM} Journal of Computing}, 40(6):1969--2000, 2011.
\newblock \href {https://doi.org/10.1137/080733644} {\path{doi:10.1137/080733644}}.

\bibitem{She12}
Alexander~A. Sherstov.
\newblock The communication complexity of gap {H}amming distance.
\newblock {\em Theory of Computing}, 8(1):197--208, 2012.
\newblock \href {https://doi.org/10.4086/toc.2012.v008a008} {\path{doi:10.4086/toc.2012.v008a008}}.

\bibitem{shor1994algorithms}
Peter~W Shor.
\newblock Algorithms for quantum computation: discrete logarithms and factoring.
\newblock In {\em Proceedings of the 35th Annual Symposium on Foundations of Computer Science}, pages 124--134, November 1994.
\newblock \href {https://doi.org/10.1109/SFCS.1994.365700} {\path{doi:10.1109/SFCS.1994.365700}}.

\bibitem{simon1994power}
Daniel~R. Simon.
\newblock On the power of quantum computation.
\newblock In {\em Proceedings of the 35th Annual Symposium on Foundations of Computer Science}, pages 116--123, November 1994.
\newblock \href {https://doi.org/10.1109/SFCS.1994.365701} {\path{doi:10.1109/SFCS.1994.365701}}.

\bibitem{8948635}
Makrand Sinha and Ronald de~Wolf.
\newblock Exponential separation between quantum communication and logarithm of approximate rank.
\newblock In {\em 2019 IEEE 60th Annual Symposium on Foundations of Computer Science (FOCS)}, pages 966--981, 2019.
\newblock \href {https://doi.org/10.1109/FOCS.2019.00062} {\path{doi:10.1109/FOCS.2019.00062}}.

\bibitem{Sur2023}
Daiki Suruga.
\newblock Matching upper bounds on symmetric predicates in quantum communication complexity.
\newblock {\em arXiv preprint}, 2023.
\newblock \href{https://arxiv.org/abs/2301.00370}{arXiv:2301.00370}.

\bibitem{Vid12}
Thomas Vidick.
\newblock A concentration inequality for the overlap of a vector on a large set, with application to the communication complexity of the {G}ap-{H}amming-{D}istance problem.
\newblock {\em Chicago Journal of Theoretical Computer Science}, 2012, 2012.
\newblock URL: \url{http://cjtcs.cs.uchicago.edu/articles/2012/1/contents.html}.

\bibitem{YZ22}
Takashi Yamakawa and Mark Zhandry.
\newblock Verifiable quantum advantage without structure.
\newblock In {\em Proceedings of the 63rd {IEEE} Annual Symposium on Foundations of Computer Science}, pages 69--74, 2022.
\newblock \href {https://doi.org/10.1109/FOCS54457.2022.00014} {\path{doi:10.1109/FOCS54457.2022.00014}}.

\bibitem{yao1979some}
Andrew~Chi{-}Chih Yao.
\newblock Some complexity questions related to distributive computing (preliminary report).
\newblock In Michael~J. Fischer, Richard~A. DeMillo, Nancy~A. Lynch, Walter~A. Burkhard, and Alfred~V. Aho, editors, {\em Proceedings of the 11h Annual {ACM} Symposium on Theory of Computing, April 30 - May 2, 1979, Atlanta, Georgia, {USA}}, pages 209--213. {ACM}, 1979.
\newblock \href {https://doi.org/10.1145/800135.804414} {\path{doi:10.1145/800135.804414}}.

\bibitem{yao1993quantum}
Andrew~Chi{-}Chih Yao.
\newblock Quantum circuit complexity.
\newblock In {\em Proceeding of the 34th Annual Symposium on Foundations of Computer Science}, pages 352--361, 1993.
\newblock \href {https://doi.org/10.1109/SFCS.1993.366852} {\path{doi:10.1109/SFCS.1993.366852}}.

\bibitem{ZS09}
Zhiqiang Zhang and Yaoyun Shi.
\newblock Communication complexities of symmetric {XOR} functions.
\newblock {\em Quantum Information and Computation}, 9(3{\&}4):255--263, 2009.
\newblock \href {https://doi.org/10.26421/QIC9.3-4-5} {\path{doi:10.26421/QIC9.3-4-5}}.

\end{thebibliography}
\end{document}